\documentclass[12pt,onecolumn,journal]{IEEEtran}
\usepackage[margin=2.54cm, lmargin=2.54cm]{geometry}
\usepackage{latexsym}
\usepackage{amsfonts}
\usepackage{amsbsy}
\usepackage{amsmath,amssymb}
\usepackage{times}
\usepackage{graphicx}
\usepackage{enumerate}
\usepackage[usenames]{color}
\usepackage[dvips]{pstcol}
\usepackage{epstopdf}
\usepackage{cite}
\usepackage{amsmath}
\usepackage{amsthm}
\usepackage{amssymb}
\usepackage{amsfonts}
\usepackage{graphicx}
\usepackage{epsfig}
\usepackage{psfrag}
\usepackage{xcolor}
\usepackage{amsfonts, bm}
\usepackage{epstopdf}
\usepackage{cite}
\usepackage{color}
\usepackage{xcolor}
\usepackage{subfig}
\usepackage{booktabs}
\usepackage{proof}
\usepackage{algorithm}
\usepackage{algorithmic}

\DeclareMathOperator*{\argmin}{arg\,min}

\input epsf
\begin{document}
\title{Low-Complexity Joint  Power Allocation and Trajectory Design for UAV-Enabled Secure Communications with Power Splitting}


\author{ Kaidi Xu, Ming-Min Zhao, Yunlong Cai, and Lajos Hanzo
\thanks{
K. Xu, M. M. Zhao, and Y. Cai are with the College of Information Science and Electronic Engineering, Zhejiang University, Hangzhou 310027, China (e-mail: xukaidi13@126.com; zmmblack@zju.edu.cn; ylcai@zju.edu.cn).

L. Hanzo is with the Department of ECS, University of Southampton, U.K. (Email: lh@ecs.soton.ac.uk).

}
}

\maketitle
\vspace{-3.3em}
\begin{abstract}
An unmanned aerial vehicle (UAV)-aided secure communication system is conceived and investigated, where the UAV transmits legitimate information to a ground user in the presence of an eavesdropper (Eve). To guarantee the security, the UAV employs a power splitting approach, where its transmit power can be divided into two parts for transmitting confidential messages and artificial noise (AN), respectively. We aim to maximize the average secrecy rate by jointly optimizing the UAV's trajectory, the transmit power levels and the corresponding power splitting ratios allocated to different time slots during the whole flight time, subject to both the maximum UAV speed constraint, the total mobility energy constraint, the total transmit power constraint, and other related constraints. To efficiently tackle this non-convex optimization problem, we propose an iterative algorithm by blending the benefits of the block coordinate descent (BCD) method, the concave-convex procedure (CCCP) and the alternating direction method of multipliers (ADMM). Specially, we show that the proposed algorithm exhibits very low computational complexity and each of its updating steps can be formulated in a nearly closed form. Our simulation results validate the efficiency of the proposed algorithm.
\end{abstract}
\vspace{-1.3em}
\begin{IEEEkeywords} \vspace{-1em}
Physical layer security, UAV, artificial noise, trajectory design, power allocation.
\end{IEEEkeywords}

\IEEEpeerreviewmaketitle

\vspace{-1em}
\section{Introduction}
\label{sec:intro}
Unmanned aerial vehicle (UAV) communications have recently attracted growing research interests in both academia and industry \cite{Zeng2016Mag, CaiJSAC2018, SunSec2019, Zhao2020UAV, Ericsson, Zeng2019, Zhang2019proceeding, Xu2019UAV}, due to many unique features and benefits, such as their prompt on-demand deployment, low latency as well as agility and flexibility. Since UAVs are generally expected to operate at a higher altitude than conventional cellular base stations (BSs), the line-of-sight (LoS) component dominates the
air-to-ground/ground-to-air channels in many practical scenarios~\cite{3GPP}. Hence UAV-aided LoS links tend to have better channel quality than typical terrestrial channels, which often suffer from severe fading and shadowing effects. However, unfortunately the UAV-aided LoS links suffer from an increased eavesdropping probability~\cite{LiangIT2008} due to the open nature of wireless channels. From this perspective, the LoS propagation of UAVs becomes a double-edged sword, since additionally the terrestrial communications are also exposed to malicious UAVs. Therefore, the delicate handling of the underlying security issues holds the key to unlocking the potential of UAV-aided communications.

Recently, physical layer security has drawn significant attention in UAV-enabled communication systems as a promising technique of protecting legitimate transmissions against eavesdropping attacks and also as a complement of conventional encryption techniques~\cite{Wu2019Sec, Xiao2018UAV}. Focusing on resource allocation/management for secrecy communication performance maximization, a range of physical layer security (PLS) techniques have been considered in the literature, such as UAV-mounted BSs~\cite{Zhang2019Secure, Li2019CL, yao2019joint, WuUAV2020}, UAV-enabled relaying~\cite{Wang2017WCL} and UAV-assisted cooperative jamming~\cite{CaiJSAC2018, Zhou2018UAV, Li2019UAV, Hua2019UAV, Chen2019UAV, Zhong2019UAV}, etc. In particular, a single-UAV communication system was investigated in~\cite{Zhang2019Secure}, where the UAV sends confidential information to a legitimate ground user (Bob) in the presence of a ground-based eavesdropper, and the secrecy rate is maximized by jointly allocating the UAV's transmit power and
optimizing its flight trajectory. The authors of~\cite{Li2019CL} have considered a scenario of multiple users and maximized the minimum secrecy rate for ensuring fairness among the users. By contrasts, the authors of~\cite{yao2019joint} considered coordinated multi-point (CoMP) reception of the legitimate users and three-dimensional (3D) trajectory optimization in the presence of multiple suspicious eavesdroppers. In~\cite{WuUAV2020}, the total transmit power of the UAV-mounted BS was minimized through joint beamforming
optimization. As a further development, the authors of \cite{Wang2017WCL} studied the security problems of UAV-aided relaying systems and judiciously allocated the transmit power levels at the source and the UAV.

Furthermore, in addition to exploiting the agile maneuverability of the UAVs for improving their secrecy performance, UAVs can also be employed as cooperative friendly jammers~\cite{Vilela2011} that are able to send artificial noise (AN) (can be viewed as external interference signals) to assist the legitimate users \cite{CaiJSAC2018, Zhou2018UAV, Li2019UAV, Hua2019UAV, Chen2019UAV, Zhong2019UAV}. Specifically, in~\cite{CaiJSAC2018}, a dual-UAV-aided secure communication scheme has been proposed, where a second UAV was employed to jam a number of eavesdroppers on the ground. In \cite{Zhou2018UAV}, the impact of the UAV's jamming power and position on the outage probability and intercept probability have been examined. In order to improve the secrecy rate, in~\cite{Li2019UAV} a mobile UAV-aided jammer was harnessed for opportunistically interfering with the potential Eve. The authors of~\cite{Hua2019UAV} studied the associated secrecy energy efficiency maximization problem, where multiple source UAVs and jamming UAVs work cooperatively to serve the ground users. In~\cite{Chen2019UAV}, AN beamforming and cooperative jamming were utilized, whilst only relying on location and statistical channel state information (CSI) of the eavesdroppers, where imperfect CSI knowledge between the UAV-aided jammer and the destination was considered. Finally, the authors of~\cite{Zhong2019UAV} considered the worst-case secrecy rate maximization problem by taking into account the uncertainty of Eve’s location.

Against the above backdrop, we investigate a UAV-enabled secure communication system, where the UAV transmits legitimate information to a ground-user Bob in the presence of a ground-based Eve. In contrast to prior studies, we conceive a power splitting aided secure transmission scheme for protecting the UAV's communications. Explicitly, the UAV divides its transmit power into two parts, where a portion $\rho$ of the signal power is used for transmitting confidential messages to Bob, while the remaining portion $1-\rho$ is devoted to transmitting AN to interfere with Eve's reception. By relying on this power splitting approach and exploiting the nimble mobility of the UAV, we aim for jointly optimizing the
trajectory of the UAV and the communicating/jamming power levels over time for maximizing the average secrecy rate of the UAV-Bob link, subject to the maximum UAV speed constraint, the total propulsion energy constraint, the total transmit power constraint, and other related constraints. To solve the resultant highly non-convex optimization problem efficiently, we propose an low-complexity iterative algorithm by combining the benefits of the block coordinate descent (BCD) method~\cite{beck2013convergence}, the concave-convex
procedure (CCCP) method~\cite{CCCP2009} and the alternating direction method of multipliers (ADMM)~\cite{boyd2011distributed}.

Specifically, in order to address the related optimization variable coupling issues, we propose to decompose the original problem into two subproblems, i.e. the power allocation subproblem and the trajectory optimization subproblem, by applying the BCD method. The resultant subproblems, although much simplified compared to the original problem, they still remain non-convex. Therefore, by exploiting the fact that the underlying non-convex parts admit a difference-of-convex (DC) structure, we propose to transform them into more tractable forms with the aid of first-order approximations. We first show that a nearly closed-form optimal solution of the approximated power allocation subproblem can be devised by resorting to its Lagrangian dual problem. Then, by tactfully introducing auxiliary variables, the approximated trajectory optimization subproblem can be iteratively and globally solved by the ADMM method, and we demonstrate that each updating step therein can also be conducted in closed-form. Given the fact that the existing algorithms suitable for solving joint power and trajectory optimization problems usually involve standard convex solvers, such as CVX \cite{cvx}, the proposed algorithm exhibits a very attractive and unique feature, namely that the optimization can be formulated almost in closed-form, thus imposing a low computational complexity. Furthermore, the proposed algorithm is proved to be monotonically convergent. Our numerical results show the benefits of the power spitting approach proposed.

The main contributions of this treatise are as follows:
 \begin{enumerate}
\item
We formulate a joint power and trajectory optimization problem for a UAV-aided secure communication system relying on a power splitting approach for improving the secrecy performance.

\item
To solve this challenging optimization problem, we propose a low-complexity iterative algorithm and show that each step in the proposed algorithm can be represented in a nearly closed form.

\item
We provide comprehensive numerical results for characterizing the efficiency of the proposed algorithm and the power splitting approach advocated. We then demonstrate the impact of the key system parameters on the average secrecy rate. In particular, we show that by appropriately splitting the transmit power of the UAV, the overall system performance can be substantially improved as compared to that without power splitting. Furthermore, compared to the existing algorithms using CVX, the running time of the proposed algorithm is at least $30$ times lower.
\end{enumerate}

This paper is structured as follows. In Section \ref{Section2:system}, we introduce the considered UAV-enabled secure communication system and formulate the joint optimization problem. In Section \ref{Section3:cccpao}, we propose an efficient iterative algorithm to solve the considered problem with very low complexity and guaranteed convergence.  Simulation results are presented in Section \ref{Section6:simulations} to show the effectiveness of our proposed algorithm and conclusions are drawn in Section
\ref{Section7:conclusion}.

\emph{Notations:} Scalars, vectors and matrices are respectively denoted by lower case, boldface lower case and boldface upper case letters. For a  matrix $\mathbf{A}$, ${{\bf{A}}^T}$ denote its transpose. For a vector $\mathbf{a}$, $\|\mathbf{a}\|$ represents its Euclidean norm. $|  \cdot  |$ denotes  the absolute value of any real or complex scalar. ${\mathbb{R}^{m \times n}}$ denotes the space of ${m \times n}$ real matrices. The set difference is defined as $\mathcal{A}\backslash \mathcal{B} \triangleq \{x| x\in\mathcal{A},x\notin \mathcal{B}\}$.  $[x]^+  \triangleq \max(x, 0)$.

\vspace{-1em}
\section{System model and Problem formulation}
In this section, we introduce the system model and formulate the optimization problem of interest. 
\label{Section2:system}
\vspace{-1em}
\subsection{System Model}
We consider a secure communication system where a UAV transmits confidential information to Bob in the presence of a potential Eve, as shown in Fig. \ref{SystemFigure}. In order to improve the security of the UAV-Bob link, the UAV also sends jamming signals (through injecting AN) to interfere Eve’s signal reception and increase the secrecy capacity.

Without loss of generality, we consider a 3D Cartesian coordinate system with Bob and Eve located at $(0, 0, 0)$ and $(L, 0, 0)$, respectively, i.e., Bob and Eve are both on the ground with a distance of $L$ meters (m). For simplicity, we focus on the UAV’s operation during a finite duration of $N$ seconds (s) and ignore its take-off and landing phases. We further assume that the UAV is flying at a fixed altitude $H$, which is considered as the minimum altitude that is required for terrain or building avoidance.\footnote{The proposed algorithm can also be extended to 3D trajectory optimization, which will become clear later.} Then, the time interval $N$ is discretized into $T$ equally spaced time slots, i.e., $N=T\delta_t$, where
$\delta_t$ denotes the elemental slot length that is chosen to be sufficiently small. Thus, the time-varying trajectory of the UAV $(x(t),y(t),H)$ over the considered time period can be approximated by the $T$-length sequence $(x[i], y[i], H)$, $i \in \mathcal{T} \triangleq \{1, \cdots, T\}$, where $(x[i], y[i])$ denotes the UAV’s $x-y$ coordinate at time slot $i$. Furthermore, let $(x_1,y_1)$ and $(x_T,y_T)$ denote the initial and final locations of the UAV and let $V_{\textrm{max}}$ denote the maximum UAV speed, then we have the following mobility constraints:
\begin{equation} \label{start_destination}
x[1]=x_1,\; y[0]=y_1,\;x[T]=x_T,\;y[T]=y_T,
\end{equation}
\begin{equation}
\label{velocity}
\sqrt{(x[i]-x[i+1])^{2}+(y[i]-y[i+1])^{2}}\leq V_{\textrm{max}}.
\end{equation}
Besides, the UAV's mobility is also constrained by its energy budget. Specifically, the energy consumed by the UAV engine at time slot $i$ is in proportion to the square of the velocity at this time slot and according to \eqref{velocity}, the energy consumed by the UAV engine at time slot $i$, denoted as $E_{\textrm{mov}}[i]$, can be expressed as \cite{JeongTVT2018, Xue2014}
\begin{equation}
\label{vp}
E_{\textrm{mov}}[i]=\kappa ((x[i]-x[i+1])^{2}+(y[i]-y[i+1])^{2}),
\end{equation}
where we have $\kappa = 0.5M \delta_t$ and $M$ denotes the UAV’s mass, including its payload. Thus, we have the following energy constraint for the mobility of the UAV:
\begin{equation}
\label{mobilitypower}
\sum_{i=1}^{T-1}E_{\textrm{mov}}[i]\leq E_{\textrm{tr}},
\end{equation}
where $E_{\textrm{tr}}$ is the total mobility energy stored at the UAV, i.e., the UAV’s energy budget.

\begin{figure}[htbp]
\centering
\includegraphics[scale=0.7]{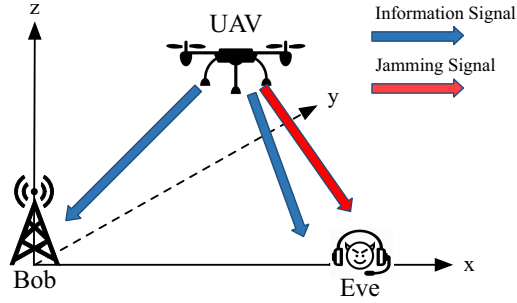}
\caption{The considered UAV-enabled secure communication system.} \vspace{-1em}
\label{SystemFigure}
\end{figure}

We assume that the LoS components dominate the channels of the UAV-Bob and UAV-Eve links, thus the channel power gains of these two links at time slot $i$ follow the frees-pace path loss model given by \cite{Zhang2019Secure, Zhao2020UAV}
\begin{equation}
g_{\textrm{I}}[i]={\gamma_{0}}/{d_{\textrm{I}}^{2}[i]}, \;
g_{\textrm{E}}[i]={\gamma_{0}}/{d_{\textrm{E}}^{2}[i]},
\end{equation}
where $\gamma_{0}$ is the power gain at the reference distance of $1$ m which depends on the carrier frequency and the antenna gains at the transmitter and receiver, $d_{\textrm{I}}[i]$ and $d_{\textrm{E}}[i]$ denote the distances from the UAV to Bob and Eve at time slot $i$, respectively, which can be expressed as
\begin{equation}
d_{\textrm{I}}[i]=\sqrt{x^2[i]+y^2[i]+H^{2}},\;
d_{\textrm{E}}[i]=\sqrt{(x[i]-L)^{2}+y^2[i]+H^{2}}.
\end{equation}
Let $p[i]$ denote the transmit power of the UAV at time slot $i$, we divide it into two parts where a portion of $p[i]\rho[i]$ is used for information transmission and the other $p[i](1-\rho[i])$ is utilized for transmitting AN to block Eve from successfully recovering the confidential information, where $\rho[i]$ is the power splitting ratio which satisfies
\begin{equation}
\label{ratio}
0\leq \rho[i]\leq 1.
\end{equation}
Note that the AN can be eliminated by Bob but not necessarily by Eve \cite{Goel2008}. The transmit power levels $\{p[i]\}$ are constrained by the limitation of both average power and peak power, which can be expressed as follows:
\begin{equation}
\label{avpower}
\frac{1}{T}\sum_{i=1}^{T}p[i]\leq\bar{P},
\end{equation}
\begin{equation}
\label{peakpower}
0\leq p[i]\leq P_{\textrm{max}},
\end{equation}
where $\bar{P}$ and $P_{\textrm{max}}$ denote the average and peak power budgets, respectively. Equivalently, the average power constraint \eqref{avpower} can be rewritten as
\begin{equation}
\label{totalpower}
\sum_{i=1}^{T}p[i]\leq P,
\end{equation}
where $P=T\bar{P}$ represents the total power available during the whole flight. Then, the signal-to-noise ratio (SNR) of the UAV-Bob link at time slot $i$ is given by
\begin{equation} \label{snri}
\textrm{SNR}_{\textrm{I}}[i] \triangleq {\gamma_{0}p[i]\rho[i]}/({d_{\textrm{I}}^2[i]\sigma^{2}}),
\end{equation}
where $\sigma^{2}$ is the additive white Gaussian noise (AWGN) variance at the receiver of Bob. Similarly, the signal-to-interference-plus-noise ratio (SINR) of the UAV-Eve link at
time slot $i$ can be expressed as
\begin{equation} \label{sinre}
\begin{aligned}
\textrm{SINR}_{\textrm{E}}[i] \triangleq  \frac{\gamma_{0}p[i]\rho[i]}{d_{\textrm{E}}^2[i] \left(\frac{\gamma_{0}(1-\rho[i])p[i]}{d_{\textrm{E}}^2[i]}+\sigma^{2}\right)} 
 = \frac{\gamma_{0}p[i]\rho[i]}{\gamma_{0}(1-\rho[i])p[i]+\sigma^{2}d_{\textrm{E}}^2[i]}.
\end{aligned}
\end{equation}
Based on \eqref{snri} and \eqref{sinre}, the secrecy rate of the UAV-Bob link at time slot $i$ is given by \cite{Gopala2008}
\begin{equation}
\label{of}
R_{\textrm{s}}[i]\triangleq [\log(1+\textrm{SNR}_{\textrm{I}}[i])-\log(1+\textrm{SINR}_{\textrm{E}}[i])]^{+},
\end{equation}
and the average secrecy rate can be written as $R_{\textrm{as}} (\{x[i],y[i],p[i],\rho[i] \})\triangleq \frac{1}{T}\sum_{i=1}^T R_{\textrm{s}}[i]$.

\vspace{-1em}
\subsection{Problem Formulation}
To this end, our objective is to maximize the average secrecy rate $R_{\textrm{as}}$ subject to the UAV’s mobility constraints in \eqref{start_destination}, \eqref{velocity} and \eqref{mobilitypower}, and the average and peak transmit power constraints in \eqref{totalpower} and \eqref{peakpower}. Therefore, we can formulate the following optimization problem:
\begin{equation}
\label{oproblem}
\begin{aligned}
	\max\limits_{\{x[i],\;y[i],\;p[i],\;\rho[i]\}}\; R_{\textrm{as}} (\{x[i],y[i],p[i],\rho[i] \})\quad
	\textrm{s.t.}\; \eqref{start_destination},\; \eqref{velocity},\; \eqref{mobilitypower},  \;\eqref{ratio},\; \eqref{peakpower}\;\textrm{and} \;\eqref{totalpower},
\end{aligned}
\end{equation}
where the optimization variables include the UAV's trajectory $\{x[i],y[i]\}$, the transmit power levels $\{p[i]\}$ and the power splitting ratios $\{ \rho[i]\}$.

Problem \eqref{oproblem} is difficult to address due to the following two reasons. First, the operator $[\cdot]^{+}$ makes the objective function of problem \eqref{oproblem} non-smooth. Second, the variables $\{p[i], \rho[i], x[i], y[i]\} $ are tightly coupled in the objective function, which makes problem \eqref{oproblem} highly non-convex. Besides, even with fixed trajectory $\{x[i],y[i]\}$ and without $[\cdot]^{+}$, the variables $\{p[i]\}$ and $\{\rho[i]\}$ are still coupled in the objective function, therefore problem \eqref{oproblem} is potentially more complex than the one considered in \cite{Zhang2019Secure}. In the next section, instead of using the existing convex solvers such as CVX \cite{cvx}, we exploit the special
structure of problem \eqref{oproblem} and propose an efficient algorithm to tackle it with low complexity by blending the benefits of the BCD method, the CCCP method and the ADMM method.

\vspace{-1em}
\section{Proposed Low-Complexity Algorithm} \label{Section3:cccpao}
First, in order to handle the non-smoothness of the objective function of \eqref{oproblem}, we can simply ignore the operator $[\cdot]^{+}$ in the objective function since if the secrecy rate is negative at an arbitrary time slot, say $l$, we can always let the corresponding transmit power $p[l]$ be $0$ such that $R_{\textrm{s}}[l]=0$ is satisfied. Therefore, ignoring the operator $[\cdot]^{+}$ causes no loss of optimality for problem \eqref{oproblem}, and we can obtain the following equivalent problem:
\begin{equation}
\label{traproblem}
\begin{aligned}
	\max\limits_{\{x[i],\;y[i],\;p[i],\;\rho[i] \}}\;  \bar{R}_{\textrm{as}} (\{x[i],y[i],p[i],\rho[i] \})\quad
	\textrm{s.t.} \; \eqref{start_destination},\; \eqref{velocity},\; \eqref{mobilitypower},  \;\eqref{ratio},\; \eqref{peakpower}\;\textrm{and} \;\eqref{totalpower},
\end{aligned}
\end{equation}
where
\begin{equation}
\label{objfun}
\begin{aligned}
\bar{R}_{\textrm{as}} (\{x[i],y[i],p[i],\rho[i] \}) \triangleq & \frac{1}{T}\sum\limits_{i=1}^{T}\Big(\log\Big(1+\frac{\gamma_{0}p[i]\rho[i]}{d_{\textrm{I}}^2[i]\sigma^{2}}\Big) \\
& -\log\Big(1+\frac{\gamma_{0}p[i]\rho[i]}{\gamma_{0}(1-\rho[i])p[i]+\sigma^{2}d_{\textrm{E}}^2[i]}\Big)\Big).
\end{aligned}
\end{equation}

Then, it can be observed that the constraints of problem \eqref{traproblem} are all convex, and the optimization variables are only coupled in the objective function. Thus, we can apply the BCD method to solve this problem by dividing the optimization variables into two blocks (i.e., $\{p[i],\rho[i] \}$ and $\{x[i], y[i]\}$) and optimizing them in an alternative manner. Specifically, with fixed trajectory, the power allocation subproblem can be expressed as
\begin{equation}
\label{pap}
\begin{aligned}
\max\limits_{\{p[i],\;\rho[i]\}}\;  \bar{R}_{\textrm{as}}(\{ p[i]\},\{\rho[i] \})\quad
\textrm{s.t.} \; \eqref{ratio},\; \eqref{peakpower}\;\textrm{and}\; \eqref{totalpower},
\end{aligned}
\end{equation}
while by fixing the transmit power levels and power splitting ratios, the trajectory optimization subproblem can be written as
\begin{equation}
\label{t_optimization_problem}
\begin{aligned}
\max\limits_{\{x[i],\;y[i]\}}\; \bar{R}_{\textrm{as}}(\{x[i],y[i] \})\quad
\textrm{s.t.} \; \eqref{start_destination},\; \eqref{velocity}\;\textrm{and}\; \eqref{mobilitypower}.
\end{aligned}
\end{equation}
In other words, we can solve problem \eqref{traproblem} by solving subproblems \eqref{pap} and \eqref{t_optimization_problem} iteratively, which will be elaborated in the following two
subsections.

\vspace{-1em}
\subsection{Solving the Power Allocation Subproblem} \label{subsection_pa}
In this subsection, we focus on problem \eqref{pap} and propose to first convert it into a convex problem through proper transformation and approximation. Then, an efficient algorithm is presented to solve the resulting convex problem by employing the Lagrange duality method, where the basic idea is to build some complicated constraints into objective functions and then solve the dual problem instead of the original problem.

To proceed, we introduce two groups of auxiliary variables $a[i]$ and $b[i]$, which satisfy
\begin{equation} \label{auxiliary}
	a[i]=p[i]\rho[i],\; b[i]=p[i](1-\rho[i]).
\end{equation}
As a result, problem \eqref{pap} can be equivalently reformulated as
\begin{subequations}
\label{tstpap}
\begin{align}
\max\limits_{\{a[i],\;b[i]\} }\; & \sum_{i=1}^Tg_i(a[i],b[i]) \label{obj_tstpap}\\
\textrm{s.t.}  \; &  a[i]+b[i]\leq P_{\textrm{max}},\;
	a[i]\geq 0,\;b[i]\geq 0, \;\forall i,\label{tstpapc1}\\
	&\sum_{i=1}^T(a[i]+b[i])\leq P,\label{tstpapc2}
\end{align}
\end{subequations}
where 
\begin{equation}
 g_i(a[i],b[i]) \triangleq \log\left(1+\frac{\gamma_{0}a[i]}{d_{\textrm{I}}^2[i]\sigma^{2}}\right)  -\log\left(1+\frac{\gamma_{0}a[i]}{\gamma_{0}b[i]+\sigma^{2}d_{\textrm{E}}^2[i]}\right).
\end{equation}
Although problem \eqref{tstpap} is much simplified as compared with problem \eqref{pap}, it is still a non-convex problem which cannot be solved efficiently in general. However, it can be readily seen that $g_i(a[i],b[i])$ can be viewed as the subtraction of two concave terms, i.e., $ \log\left(1+\frac{\gamma_{0}a[i]}{d_{\textrm{I}}^2[i]\sigma^{2}}\right)
+ \log\left(\gamma_{0}b[i]+\sigma^{2}d_{\textrm{E}}^2[i]\right)$ and $\log\left(\gamma_{0}b[i]+\sigma^{2}d_{\textrm{E}}^2[i]+\gamma_{0}a[i]\right) $, or equivalently, \eqref{obj_tstpap} can be expressed in a DC form. Therefore, by employing the CCCP method\cite{CCCP2009, Zhao2017TCOM, Zhao2017TWC}, the lower bound of \eqref{obj_tstpap} can be obtained as
\begin{equation} \label{appro}
\sum\limits_{i=1}^T g_i(a[i],b[i]) \geq \sum\limits_{i=1}^T \hat{g}_i(a[i],b[i];a_f[i],b_f[i]),
\end{equation}
where $\{a_f[i],b_f[i]\}$ is the given feasible solution of problem \eqref{pap}\footnote{In the following, the subscript $f$ is used to denote the feasible variable obtained in the previous BCD iteration.} and
\begin{equation} \label{xxx}
\begin{aligned}
	\hat{g}_i(a[i],b[i];a_f[i],b_f[i])\triangleq  \log\left(1+\frac{\gamma_{0}a[i]}{d_{\textrm{I}}^2[i]\sigma^{2}}\right)
   	- \log\left(\gamma_{0}b_f[i]+\sigma^{2}d_{\textrm{E}}^2[i]+\gamma_{0}a_f[i]\right) \\
   	+ \log\left(\gamma_{0}b[i]+\sigma^{2}d_{\textrm{E}}^2[i]\right)
 	- \frac{\gamma_0}{\gamma_{0}b_f[i]+\sigma^{2}d_{\textrm{E}}^2[i]+\gamma_{0}a_f[i]}(a[i]-a_f[i]+b[i]-b_f[i]).
\end{aligned}
\end{equation}
Note that the equality in \eqref{appro} holds when $a[i]=a_f[i]$ and $b[i]=b_f[i]$. Consequently, problem \eqref{tstpap} can be approximated by the following convex problem:
\begin{equation}
\label{pap1}
\begin{aligned}
	\max\limits_{\{a[i],\;b[i] \}}\; \sum\limits_{i=1}^T	\hat{g}_i(a[i],b[i];a_f[i],b_f[i])\quad
	\textrm{s.t.}  \; \eqref{tstpapc1} \; \textrm{and}\; \eqref{tstpapc2}.
\end{aligned}
\end{equation}

Then, we note that without the total power constraint \eqref{tstpapc2}, the other constraints in problem \eqref{pap1} are separable over different time slots $i \in \mathcal{T}$. Inspired by this observation, we introduce a Lagrange multiplier (dual variable) $\lambda \geq 0$ to \eqref{tstpapc2} and define the partial Lagrangian associated with problem \eqref{pap1} as \cite{ConvexOptimization}
\begin{equation} \label{patial_lagrangian}
\begin{aligned}
\mathcal{L}(\{a[i]\},\{b[i]\},\lambda)=  \sum\limits_{i=1}^T	\hat{g}_i(a[i],b[i];a_f[i],b_f[i]) -\lambda\sum\limits_{i=1}^{T}(a[i]+b[i])+\lambda P.
\end{aligned}
\end{equation}
With \eqref{patial_lagrangian}, the dual function, denoted by $d(\lambda)$, can be written as \cite{ConvexOptimization}
\begin{equation}
\label{tpap}
\begin{aligned}
d(\lambda) \triangleq  \max\limits_{\{a[i],\;b[i]\}}\; \mathcal{L}(\{a[i]\},\{b[i]\},\lambda) \quad
\textrm{s.t.} \;  \eqref{tstpapc1}.
\end{aligned}
\end{equation}
Let $\{a[i](\lambda)\}$ and $\{b[i](\lambda)\}$ denote an optimal solution of problem \eqref{tpap} with fixed $\lambda$.  It is not difficult to see that, if $\{a[i](0),b[i](0)\}$ satisfy the total
power constraint \eqref{tstpapc2}, then $\{a[i](0),b[i](0)\}$ is optimal for problem \eqref{pap1}, since when $\lambda =0$, problem \eqref{tpap} becomes a relaxed version of problem \eqref{pap1} without the total power constraint \eqref{tstpapc2} and if \eqref{tstpapc2} is automatically satisfied in this case, the only possibility is that $\{a[i](0),b[i](0)\}$ is optimal. Otherwise, we need to increase $\lambda$ to enhance the dominance of $-\lambda\sum_{i=1}^{T}(a[i]+b[i])+\lambda P$ in $\mathcal{L}(\{a[i]\},\{b[i]\},\lambda)$ and force $\{a[i](\lambda),b[i](\lambda)\}$ to satisfy \eqref{tstpapc2}.

Since problem \eqref{pap1} is convex and strong duality \cite{ConvexOptimization} holds, we have $p^{\textrm{opt}} = d(\lambda^{\textrm{opt}}) \leq d(\lambda)$ for any $\lambda \geq 0$,
where $p^{\textrm{opt}}$ is the optimal objective value of problem \eqref{pap1} and $\lambda^{\textrm{opt}}$ denotes the optimal dual variable. Hence, in order to solve problem \eqref{pap1}, we can instead solve the following dual problem:
\begin{equation} \label{dual_problem}
\min\limits_{\lambda \geq 0}\; d(\lambda).
\end{equation}
Since $d(\lambda)$ is a convex function with respect to $\lambda$ and $P -\sum_{i=1}^{T}(a[i]+b[i])$ is a subgradient of $d(\lambda)$ \cite[pp. 12]{boyd2007notes}, we can infer that if $\{
a[i](\lambda^{\textrm{opt}}),b[i](\lambda^{\textrm{opt}})\} $ satisfies \eqref{tstpapc2} and $\lambda^{\textrm{opt}}\big(\sum_{i=1}^T
(a[i](\lambda^{\textrm{opt}})+b[i](\lambda^{\textrm{opt}}))- P\big) = 0$, then $\{ a[i](\lambda^{\textrm{opt}}),b[i](\lambda^{\textrm{opt}})\}$ is an optimal solution of problem \eqref{pap1}.

To this end, our main focus is on solving the dual problem \eqref{dual_problem} and this can be conducted by using the Bisection method \cite{ConvexOptimization} with the aid of the subgradient $P -\sum_{i=1}^{T}(a[i]+b[i])$. We summarize the proposed Lagrange duality method in Algorithm \ref{tab:table2}, where Steps 1-4 check whether or not $\{a[i](0),b[i](0)\}$ is the optimal solution, Steps 5-15 represent the Bisection method to solve the dual problem \eqref{dual_problem} globally. Note that in Steps 10-14, we increase $\lambda$ when the subgradient $P -\sum_{i=1}^{T}(a[i]+b[i])$ is positive and decrease $\lambda$ otherwise, so as to find the optimal dual variable.  In the following, we show that problem \eqref{tpap} can be solved globally in closed-form with given $\lambda$.

\begin{algorithm}[!h] \small
	\caption{Proposed Algorithm for Solving Problem \eqref{pap1}} 
\begin{algorithmic}[1] \label{tab:table2}
			\STATE Let $\lambda \gets 0$ and solve problem \eqref{tpap} to obtain $\{a[i](0),b[i](0)\}$.
			\IF {$\sum_{i=1}^{T}(a[i](0)+b[i](0)) \leq P$}
			\STATE \textbf{output} $\{a[i](0), b[i](0)\}$ and exit the algorithm.
			\ENDIF
				\STATE $\lambda_l \gets 0$, find $\lambda_r$ such that that $\sum_{i=1}^{T}(a[i](\lambda_r)+b[i](\lambda_r)) \leq P$.
			\REPEAT
			\STATE $\lambda \gets ({\lambda_l+\lambda_r})/{2}$.
			\STATE Obtain $\{a[i](\lambda),b[i](\lambda)\}$ by solving problem \eqref{tpap}.
			\STATE \textbf{if} {$\sum_{i=1}^{T}(a[i](\lambda)+b[i](\lambda)) < P$} \textbf{then} $\lambda_r \gets \lambda$, \textbf{else} $\lambda_l \gets \lambda$. \textbf{end if}
			\UNTIL{$|\sum_{i=1}^{T}(a[i](\lambda)+b[i](\lambda)) - P|$ is less than a certain threshold.}
			\STATE \textbf{output} ($\{a[i](\lambda), b[i](\lambda)\}$.
	\end{algorithmic}
\end{algorithm}

It is readily seen that problem \eqref{tpap} can be divided into $T$ independent subproblems for each time slot $i$. Since each subproblem can be solved similarly, we only need to focus on one particular subproblem, and the corresponding optimization problem can be expressed as (the time slot index is omitted here for simplicity)
\begin{equation}\label{stpap}
\begin{aligned}
	\max\limits_{a,\;b}\;  \tilde{g}(a,b) \quad 
	\textrm{s.t.} \;  a+b\leq P_{\textrm{max}},\;
	a\geq 0,\;b\geq 0,
\end{aligned}
\end{equation}
where $\tilde{g}(a,b)\triangleq \hat{g}(a,b;a_f,b_f)-\lambda(a+b)$. It can be observed that problem \eqref{stpap} is convex and there are only two optimization variables.  With fixed $a$, $ \tilde{g}(a,b)$ is a strictly concave function with respect to $b$ since $\log(1+x)$ ($x\geq 0$) is strictly concave. In what follows, we show how problem \eqref{stpap} can be efficiently solved with low complexity.

First, we recast problem \eqref{stpap} as the following equivalent two-tier maximization problem:
\begin{equation} \label{two_tier_problem}
\max\limits_{0\leq a \leq P_{\textrm{max}}-b}\; \max\limits_{0\leq b\leq  P_{\textrm{max}}-a}\; \tilde{g}(a,b).
\end{equation}
For given $a$, the optimal $b$ (it is unique since $\tilde{g}(a,b)$ is strictly concave with fixed $a$), denoted as $\bar{b}(a)$, can be obtained by resorting to the first-order optimality condition of the
inner maximization problem, i.e.,
\begin{equation}
	\frac{d \tilde{g}(a,b)}{d b}=\frac{\gamma_0}{\gamma_0b+\sigma^2 d_{\textrm{E}}^2} - \frac{\gamma_0}{\gamma_0b_f+\sigma^2d_{\textrm{E}}^2+\gamma_0a_f} -\lambda=0,
\end{equation}
and we can obtain the stationary point of $\tilde{g}(a,b)$ as $b_s={1}/{C_b} - {\sigma^2d_{\textrm{E}}^2}/{\gamma_0}$, where $C_b\triangleq \lambda + \frac{\gamma_0}{\gamma_0b_f+\sigma^2d_{\textrm{E}}^2+\gamma_0a_f}$.

Since the inner maximization problem is a univariate convex problem with a bound constraint, its optimal objective value must be attained either on the boundary of the constraint or at the stationary point $b_s$. To be specific, the optimal solution of the inner maximization problem can be obtained by
\begin{equation}
\label{bopt}
	\bar{b}(a)=\left\{
	\begin{array}{l}
	 0 ,\; \textrm{if}\; b_s\leq 0,\\
	 b_s ,\; \textrm{if }\; 0< b_s< P_{\textrm{max}}-a,\\
	 P_{\textrm{max}}-a,\;\textrm{otherwise}.
	\end{array}
	\right.
\end{equation}
Substituting $ \bar{b}(a)$ into the objective function of the outer maximization problem of \eqref{two_tier_problem}, it can be recast as follows with $a$ as the only variable:
\begin{equation}
\label{atstpap}
\begin{array}{l}
	\max\limits_{a}\;\bar{g}(a)\quad
	\textrm{s.t.}   \; a+\bar{b}(a)\leq P_{\textrm{max}},\; a\geq 0,
\end{array}
\end{equation}
where $\bar{g}(a) \triangleq \tilde{g}(a,\bar{b}(a))$.  As discussed above, for a univariate optimization problem with a bound constraint, the optimal objective value must be attained at either the endpoints of the bound interval or some feasible stationary point of the objective function. Accordingly, the optimal value of problem \eqref{atstpap} must be attained either at the point that satisfies $\frac{d\bar{g}(a)}{da}=0 \;( 0 < a < P_{\textrm{max}}-\bar{b}(a))$, or $a\in \{0, P_{\textrm{max}}-\bar{b}(a)\}$.  Therefore, our basic idea to solve problem \eqref{atstpap} is to search over all stationary points and boundary points and then choose the one that achieves the maximum objective value.

Next, we solve problem \eqref{atstpap} by considering the above mentioned two cases. By taking the derivative of $\bar{g}(a)$ with respect to $a$, we have
\begin{equation}
\label{daobj}
\begin{aligned}
\frac{d\bar{g}(a)}{da}= \frac{\gamma_0}{\gamma_0a+\sigma^2d_{\textrm{I}}^2} + \frac{\gamma_0\frac{d\bar{b}(a)}{da}}{\gamma_0\bar{b}(a)+\sigma^2d_{\textrm{E}}^2} - \Big(\frac{\gamma_0}{\gamma_0b_f+\sigma^2d_{\textrm{E}}^2+\gamma_0a_f}+\lambda\Big)(1+\frac{d\bar{b}(a)}{da}).\\
\end{aligned}
\end{equation}
\subsubsection{\textbf{Case I} ($0 < a <  P_{\textrm{max}}-\bar{b}(a)$)} According to \eqref{bopt},  we need to further consider the following two cases: $\bar{b}(a)=b_s$ or $\bar{b}(a)=0$. For both cases, we have $\frac{d\bar{b}(a)}{da}=0$. By plugging $\frac{d\bar{b}(a)}{da}=0$ into \eqref{daobj} and letting \eqref{daobj} equal to $0$, we have
$\frac{\gamma_0}{\gamma_0a+\sigma^2d_{\textrm{I}}^2} =
\frac{\gamma_0}{\gamma_0b_f+\sigma^2d_{\textrm{E}}^2+\gamma_0a_f}+\lambda$. Accordingly,
$a$ can be obtained by
\begin{equation}
	 a=\frac{\gamma_0b_f+\sigma^2d_{\textrm{E}}^2+\gamma_0a_f}{{\gamma_0 + \lambda(\gamma_0b_f+\sigma^2d_{\textrm{E}}^2+\gamma_0a_f)}}-\frac{d_{\textrm{I}}^2\sigma^2}{\gamma_0}.
\end{equation}

\subsubsection{\textbf{Case II} ($a\in \{0, P_{\textrm{max}}-\bar{b}(a)\}$)} In this case, $a$ can take on two possible values, i.e., $a=0$ or $a=P_{\textrm{max}}-\bar{b}(a)$. If $a=0$, we have $b=\bar{b}(0)$, otherwise, if $a=P_{\textrm{max}}-\bar{b}(a)$, this implies that $\bar{b}(a)=P_{\textrm{max}}-a$ and $\frac{d\bar{b}(a)}{da}=-1$. Consequently, we have
$\frac{d\bar{g}(a)}{da}=\frac{\gamma_0}{d_{\textrm{I}}^2\sigma^2+\gamma_0a} -\frac{\gamma_0}{d_{\textrm{E}}^2\sigma^2+\gamma_0(P_{\textrm{max}}-a)}=0$, which can be further simplified to a linear equation and its solution can be easily obtained by
\begin{equation}
 a={(\sigma^2(d_{\textrm{E}}^2-d_{\textrm{I}}^2)+\gamma_0P_{\textrm{max}})}/{(2\gamma_0)}.
\end{equation}
Then, by checking the abovementioned four sub-cases and discarding those do not satisfy the case conditions $0 < a < P_{\textrm{max}}-\bar{b}(a)$ or $a\in \{0, P_{\textrm{max}}-\bar{b}(a)\}$, we can obtain several feasible solutions of problem \eqref{stpap}. Consequently, problem \eqref{stpap} can be globally solved in closed-form by choosing the feasible solution that achieves the maximum objective value.

Together with Algorithm \ref{tab:table2}, the approximated power allocation subproblem \eqref{pap1} can be efficiently solved and with the optimized $\{a[i] ,b[i]\}$, we can easily obtain $\{p[i],\rho[i] \}$ according to \eqref{auxiliary}.

\vspace{-1em}
\subsection{Solving the Trajectory Optimization Subproblem}
In this subsection, we focus on solving the trajectory optimization subproblem \eqref{t_optimization_problem} with fixed $\{p[i], \rho[i] \}$. Note that although the constraints of problem \eqref{t_optimization_problem} are convex, its objective function is non-concave with respect to $\{x[i], y[i]\}$ and it cannot be solved optimally in general. In order to resolve the difficulty caused by the non-concave objective function, we introduce two sets of auxiliary variables  $\{u[i]\}$ and $\{t[i]\}$, which satisfy
\begin{equation} \label{uc}
	u[i]\geq x^{2}[i]+y^{2}[i]+H^{2},
\end{equation}
\begin{equation} \label{tc}
	t[i]\leq (x[i]-L)^{2}+y^{2}[i]+H^{2}.
\end{equation}
As a result, we have the following equivalent optimization problem:
\begin{equation}
\label{tproblem_equi}
\begin{aligned}
\max\limits_{\{x[i],\;y[i],\;u[i],\;t[i]\}}\;  \tilde{R}_{\textrm{as}}(u[i],t[i])\quad
\textrm{s.t.} \; \eqref{start_destination},\; \eqref{velocity},\; \eqref{mobilitypower},\; \eqref{uc}\;\textrm{and}\; \eqref{tc},
\end{aligned}
\end{equation}
where 
\begin{equation}
\label{TTASR}
\begin{aligned}
\tilde{R}_{\textrm{as}}(u[i],t[i])\triangleq   \sum\limits_{i=1}^{T}\Big(\log\Big(1+\frac{\gamma_{0}p[i]\rho[i]}{u[i]\sigma^{2}}\Big)
 - \log\Big(1+\frac{\gamma_{0}p[i]\rho[i]}{\gamma_{0}(1-\rho[i])p[i]+\sigma^{2}t[i] }\Big)\Big).
\end{aligned}
\end{equation}
We note that constraints \eqref{uc} and \eqref{tc} in problem \eqref{tproblem_equi} must be satisfied with equality at optimality since otherwise, we can always slightly decrease $u[i]$ and increase $t[i]$ such that a larger objective value can be achieved without violating any constraint. Therefore, problem \eqref{t_optimization_problem} and problem \eqref{tproblem_equi} are equivalent.

It can be observed that the term $\log\Big(1+\frac{\gamma_{0}p[i]\rho[i]}{u[i]\sigma^{2}}\Big)$ in $\tilde{R}_{\textrm{as}}(u[i],t[i])$ and the term $(x[i]-L)^{2}+y^{2}[i]$ in \eqref{tc} are convex with respect to $u[i]$ and $\{x[i], y[i]\}$, respectively. Therefore, although $\tilde{R}_{\textrm{as}}(u[i],t[i])$ is non-concave and constraint \eqref{tc} is non-convex, they can be expressed in DC forms and problem \eqref{tproblem_equi} can be addressed by employing the CCCP method.  Specifically, we propose to approximate problem \eqref{tproblem_equi} to a convex one and then present an ADMM-based algorithm to solve it globally. First, the proposed algorithm assumes a given solution $\{x_f[i],y_f[i], u_f[i],t_f[i] \}$ in the previous BCD iteration which is feasible to \eqref{tproblem_equi}. Then, by employing the first-order Taylor approximation, we construct the lower bounds for $ (x[i]-L)^{2}+y^{2}[i]+H^{2}$ and
$\log\left(1+\frac{\gamma_{0}p[i]\rho[i]}{u[i]\sigma^{2}}\right)$ as follows:
\begin{equation}
\begin{aligned}
	 -x_{f}^{2}[i]+ 2x_{f}  [i]x[i]-2x[i]L+L^{2}-y_{f}^{2}[i]
	+2y_{f}[i]y[i]+H^{2}\leq (x[i]-L)^{2}+y^{2}[i]+H^{2},
	\end{aligned}
\end{equation}
\begin{equation}
\begin{aligned}
\log\left(1+\frac{\gamma_{0}p[i]\rho[i]}{u[i]\sigma^{2}}\right)
\geq   \log(1+\frac{\gamma_{0}p[i]\rho[i]}{u_{f}[i]\sigma^{2}})
 -\frac{\gamma_{0}p[i]\rho[i](u[i]-u_{f}[i])}{u_{f}^{2}[i]\sigma^{2}+\gamma_{0}p[i]\rho[i]u_{f}[i]}.
\end{aligned}
\end{equation}
Similarly, we also approximate the second term in \eqref{TTASR}, i.e., $\log\Big(1+\frac{\gamma_{0}p[i]\rho[i]}{\gamma_{0}(1-\rho[i])p[i]+\sigma^{2}t[i] }\Big)$, and obtain the following upper bound:
\begin{equation}\label{lower_bound}
\begin{aligned}
\log\left(1+\frac{\gamma_{0}p[i]\rho[i]}{\gamma_{0}(1-\rho[i])p[i]+\sigma^{2}t[i]}\right) = \log(\gamma_{0}p[i]+\sigma^{2}t[i]) 
-  \log\left(\gamma_{0}(1-\rho[i])p[i]+\sigma^{2}t[i]\right)\\
 \leq \log(\gamma_{0}p[i]+\sigma^{2}t_{f}[i])
+\frac{\sigma^{2}(t[i]-t_{f}[i])}{t_{f}[i]\sigma^{2}+\gamma_{0}p[i]}
-\log(\gamma_{0}p[i](1-\rho[i])+\sigma^{2}t[i]).
\end{aligned}
\end{equation}
Note that although replacing $\log\Big(1+\frac{\gamma_{0}p[i]\rho[i]}{\gamma_{0}(1-\rho[i])p[i]+\sigma^{2}t[i] }\Big)$ by its upper bound in \eqref{lower_bound} is mathematically unnecessary since it is already a convex function, it will be clear later that with this approximation, the resulting problem is easier to handle. Moreover, we will show in the simulation results that even with such additional approximation, the performance achieved by the proposed low-complexity algorithm is similar to that achieved by using the CVX solver. After the above mentioned approximations, it is not difficult to see that the original non-concave objective function $\tilde{R}_{\textrm{as}}(u[i],t[i])$ and non-convex constraint \eqref{tc} in problem \eqref{tproblem_equi} can be approximated by
\begin{equation}
	\label{ctc}
	\begin{aligned}
	t[i]\leq -x_{f}^{2}[i]+2x_{f}[i]x[i]-2x[i]L+L^{2}-y_{f}^{2}[i] +2y_{f}[i]y[i]+H^{2},
	\end{aligned}
\end{equation}
\begin{equation}
\label{LTTASR}
\begin{aligned}
 \check{R}_{\textrm{as}} (u[i],t[i])\triangleq & \sum\limits_{i=1}^T \Big( -\frac{\gamma_{0}p[i]\rho[i](u[i]-u_{f}[i])}{u_{f}^{2}[i]\sigma^{2}+\gamma_{0}p[i]\rho[i]u_{f}[i]}\\
& -\frac{\sigma^{2}(t[i]-t_{f}[i])}{t_{f}[i]\sigma^{2}+\gamma_{0}p[i]} +\log(\gamma_{0}p[i](1-\rho[i])+\sigma^{2}t[i]) \Big),
\end{aligned}
\end{equation}
respectively.\footnote{Note that in \eqref{LTTASR}, some constant terms are ignored for simplicity.} Therefore, problem \eqref{tproblem_equi} can be approximated as the following convex problem:
\begin{equation}
\label{tproblem_appro}
\begin{aligned}
	\max\limits_{\{x[i],\;y[i],\;u[i],\;t[i]\}}\; \check{R}_{\textrm{as}}(u[i],t[i])\quad 
	\textrm{s.t.} \; \eqref{start_destination},\; \eqref{velocity},\; \eqref{mobilitypower},\;\eqref{uc}\;\textrm{and}\; \eqref{ctc}.
\end{aligned}
\end{equation}

Subsequently, we develop a low-complexity ADMM-based algorithm to globally solve problem \eqref{tproblem_appro} efficiently.  By exploiting the special structure of problem \eqref{tproblem_appro}, we show that by tactfully introducing auxiliary variables, it can be efficiently solved and each step in the proposed ADMM method can be carried out in closed-form and in parallel. For completeness, a brief introduction of the ADMM method is provided in Appendix \ref{appendixA}.  It can be seen that problem \eqref{tproblem_appro}
is not in the standard form of problem \eqref{ADMM}, therefore, it is difficult to directly apply the ADMM method.  The main difficulties lie in: 1) how to partition the optimization variables of problem \eqref{tproblem_appro} into two groups, as in the ADMM framework, 2) how to decompose each group problem for much easier implementation. To proceed, we introduce four redundancy copies of the variables $\{x[i], y[i]\}$ to help address the abovementioned difficulties, i.e.,
\begin{subequations} \label{eq_cons}
\begin{align}
x[i]=\bar{x}[i] ,\; y[i]=\bar{y}[i],\;
	x[i]=\tilde{x}[i] ,\; y[i]=\tilde{y}[i],\\
	x[i]=\hat{x}[i] ,\; y[i]=\hat{y}[i], \;
	\hat{x}[i]=\ddot{x}[i] ,\; \hat{y}[i]=\ddot{y}[i].
\end{align}
\end{subequations}
Then, due to the introduction of \eqref{eq_cons}, constraints \eqref{velocity}, \eqref{mobilitypower}, \eqref{uc} and \eqref{ctc} are modified as follows without loss of optimality:
 \begin{equation} \label{velocity2}
 (x[i]-\bar{x}[i+1])^{2}+(y[i]-\bar{y}^{2}[i+1])\leq  V_{\textrm{max}}^2,
 \end{equation}
 \begin{equation} \label{mobilitypower2}
 \sum\limits_{n=1}^{T-1} \big( (\ddot{x}[i]-\ddot{x}[i+1])^2+(\ddot{y}[i]-\ddot{y}[i+1])^2\big) \leq \frac{E_{\textrm{tr}}}{\kappa},
 \end{equation}
 \begin{equation} \label{uc2}
 u[i]\geq \tilde{x}^{2}[i]+\tilde{y}^{2}[i]+H^{2},
 \end{equation}
 \begin{equation}
 \label{ctc2}
 \begin{aligned}
 t[i]\leq -\tilde{x}_{f}^{2}[i]+2\tilde{x}_{f}[i]x[i] -2\tilde{x}[i]L+L^{2}-\tilde{y}_{f}^{2}[i] +2\tilde{y}_{f}[i]\tilde{y}[i]+H^{2}.
 \end{aligned}
 \end{equation}
 
 Next, by dualizing and penalizing the equality constraints in  \eqref{eq_cons} to the objective function, we can obtain the augmented Lagrangian (AL) function of problem \eqref{tproblem_appro}, which is given by
\begin{equation}
	\label{AL}
	\begin{aligned}
	& L_{\delta}(\mathcal{Q},\mathcal{U})= \check{R}_{\textrm{as}}(u[i],t[i])	-\frac{\delta}{2} \sum\limits_{i=1}^T\Big((x[i]-\bar{x}[i]-\frac{\lambda_{x_i}}{\delta})^{2}+(y[i]-\bar{y}[i]-\frac{\lambda_{y_i}}{\delta})^{2}\\	
	& +(x[i]-\tilde{x}[i]-\frac{\eta_{x_i}}{\delta})^{2}+(y[i]-\tilde{y}[i]-\frac{\eta_{y_i}}{\delta})^{2} +(x[i]-\hat{x}[i]-\frac{\omega_{x_i}}{\delta})^{2}+(y[i]-\hat{y}[i]-\frac{\eta_{y_i}}{\delta})^{2}\\
	& +(\hat{x}[i]-\ddot{x}[i]-\frac{\theta_{x_i}}{\delta})^{2}+(\hat{y}[i]-\ddot{y}[i]-\frac{\theta_{y_i}}{\delta})^{2}\Big),
	\end{aligned}
\end{equation}
where $\mathcal{Q} \triangleq \{ x[i],y[i],\bar{x}[i],\bar{y}[i],\hat{x}[i],\hat{y}[i],\tilde{x}[i],\tilde{y}[i],\ddot{x}[i], \ddot{y}[i],$ $u[i],t[i]\}$, $\delta$ is the penalty parameter, $\mathcal{U} \triangleq \{\lambda_{x_i},\lambda_{y_i},\eta_{x_i},\eta_{y_i},\omega_{x_i},\omega_{y_i}, \theta_{x_i},\theta_{y_i}\}$, $\{\lambda_{x_i},\lambda_{y_i}\}$, $\{\eta_{x_i},\eta_{y_i}\}$, $\{\omega_{x_i},\omega_{y_i}\}$ and $\{\theta_{x_i},\theta_{y_i}\}$ are the dual variables associated with the constraints in \eqref{eq_cons}, respectively. Accordingly, we have the following AL problem:
\begin{equation}
\label{tproblem}
\begin{aligned}
	\max\limits_{\mathcal{Q}}\;  L_{\delta}(\mathcal{Q},\mathcal{U})\quad
	\textrm{s.t.} \; \eqref{start_destination},\;\eqref{eq_cons}-\eqref{ctc2}.
\end{aligned}
\end{equation}

To solve problem \eqref{tproblem}, we need to divide the primal variables $\mathcal{Q}$ into two groups (correspond to $\mathbf{x}$ and $\mathbf{z}$ in Appendix \ref{appendixA}). For this purpose and to facilitate parallel implementation, we group the variables $\mathcal{Q}\backslash \{ \ddot{x}[i], \ddot{y}[i]\}$ according to the parity of their corresponding time slot indices, while the variables $\{\ddot{x}[i], \ddot{y}[i]\}$ are handled in one group since they all appear in constraint \eqref{mobilitypower2}. Besides, we also classify these variables into three different types according to the forms of their corresponding optimization subproblems, as shown in Fig. \ref{variable_grouping}. In the following, we elaborate the details on how to solve these subproblems efficiently.
\begin{figure}[htbp]
	\centering
	\includegraphics[scale=0.3]{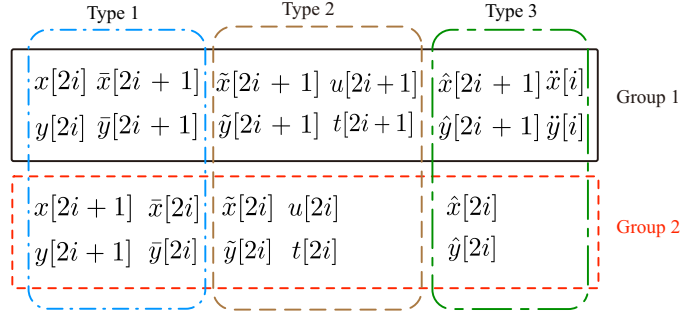}
	\caption{Grouping and classification of the optimization variables.} \vspace{-1em}
	\label{variable_grouping}
\end{figure}

\subsubsection{Group 1}
The Type 1 subproblem is involved with variables $\{x[2i], \bar{x}[2i+1], y[2i], \bar{y}[2i+1]\}$ and the corresponding optimization problem can be expressed as
\begin{equation}
\label{sb1}
\begin{array}{l}
	\max\limits_{x[2i],\; \bar{x}[2i+1],\; y[2i],\; \bar{y}[2i+1]}\; L_{\delta,1} \quad
	\textrm{s.t.} \; (x[2i]-\bar{x}[2i+1])^{2}+(y[2i]-\bar{y}[2i+1])^{2}\leq  V_{\textrm{max}}^2,
\end{array}
\end{equation}
where
\begin{equation}
\begin{aligned}
	L_{\delta,1} \triangleq& -\frac{\delta}{2}\Big(x[2i]-\bar{x}[2i]-\frac{\lambda_{x_{2i}}}{\delta})^{2}+(y[2i]-\bar{y}[2i]-\frac{\lambda_{y_{2i}}}{\delta})^{2}\\
&	+(x[2i]-\tilde{x}[2i]-\frac{\eta_{x_{2i}}}{\delta})^{2}+(y[2i]-\tilde{y}[2i]-\frac{\eta_{y_{2i}}}{\delta})^{2}\\	
& +(x[2i]-\hat{x}[2i]-\frac{\omega_{x_{2i}}}{\delta})^{2}+(y[2i]-\hat{y}[2i]-\frac{\eta_{y_{2i}}}{\delta})^{2}\\
&	+(x[2i+1]-\bar{x}[2i+1]-\frac{\lambda_{x_{2i+1}}}{\delta})^{2} +(y[2i+1]-\bar{y}[2i+1]-\frac{\lambda_{y_{2i+1}}}{\delta})^{2}\Big).
	\end{aligned}
\end{equation}
Problem \eqref{sb1} is a quadratically constrained quadratic programming (QCQP) problem with only one constraint, therefore, it can be globally solved and the detailed derivation of its optimal solution is relegated to Appendix \ref{appendixB}.  Note that for each time slot $i$, the corresponding variables can be optimized in parallel.

The Type 2 subproblem involves the optimization of $\{\tilde{x}[2i+1],\tilde{y}[2i+1],u[2i+1],t[2i+1]\}$, which can be written as
\begin{equation}
\label{sb3}
\begin{array}{l}
\max\limits_{\tilde{x}[2i+1],\; \tilde{y}[2i+1],\; u[2i+1],\; t[2i+1]}\;  L_{\delta,2} \\
\textrm{s.t.} \; u[2i+1]\geq \tilde{x}^2[2i+1]+\tilde{y}^2[2i+1]+H^2,\\
\quad\;\; t[2i+1]\leq -\tilde{x}_f^2[2i+1]+2\tilde{x}_f[2i+1]\tilde{x}[2i+1] +L^2 -2\tilde{x}[2i+1]L-\tilde{y}_f^2[2i+1]\\
\quad\quad\quad\quad \quad \quad\; +2\tilde{y}_f[2i+1]\tilde{y}[2i+1]+H^2,
\end{array}
\end{equation}
where
\begin{equation}
\begin{aligned}
L&_{\delta,2} \triangleq -a[2i+1]u[2i+1]-\frac{\sigma^2t[2i+1]}{t_f[2i+1]+\gamma_0p[2i+1]}\\
&+\log(\gamma_0(1-\rho[2i+1])p[2i+1]+\sigma^2t[2i+1]) -\frac{\delta}{2}\big((x[2i+1]-\tilde{x}[2i+1]-\frac{\omega_{x_{2i+1}}}{\delta})^2\\
&+(y[2i+1]-\tilde{y}[2i+1]-\frac{\omega_{y_{2i+1}}}{\delta})^2\big),
\end{aligned}
\end{equation}
\begin{equation}
a[2i+1]=\frac{\gamma_0p[2i+1]\rho[2i+1]}{u_f^2[2i+1]\sigma^2+\gamma_0p[2i+1]\rho[2i+1]u_f[2i+1]}.
\end{equation}
It can be observed that problem \eqref{sb3} is a QCQP problem with two constraints. Although there is no closed-form solution for such kind of optimization problems in general, we show that it can be efficiently solved in closed-form by exploiting its special structure and the details are provided in Appendix \ref{appendixC}.

The Type 3 subproblem involves the optimization of $\{ \hat{x}[2i+1],\hat{y}[2i+1],\ddot{x}[i],\ddot{y}[i]\}$ and we can obtain the following problem:
\begin{equation}\label{sb21}
\begin{aligned}
\max\limits_{\{\hat{x}[2i+1],\; \hat{y}[2i+1], \;\ddot{x}[i],\; \ddot{y}[i]\}} \;  L_{\delta,3} \quad
\textrm{s.t.}\; \sum\limits_{n=1}^{T-1} \big( (\ddot{x}[i]-\ddot{x}[i+1])^2+(\ddot{y}[i]-\ddot{y}[i+1])^2\big) \leq \frac{E_{\textrm{tr}}}{\kappa},
\end{aligned}
\end{equation}
where
\begin{equation}
\begin{aligned}
L_{\delta,3} \triangleq & -\frac{\delta}{2}\sum\limits_{n=1}^T\Big(\hat{x}[i]-\ddot{x}[i]-\frac{\theta_{x_i}}{\delta})^2+(\hat{y}[i]-\ddot{y}[i]-\frac{\theta_{y_i}}{\delta})^2\\
& +(x[i]-\hat{x}[i]-\frac{\eta_{x_i}}{\delta})^2+(y[i]-\hat{y}[i]-\frac{\eta_{y_i}}{\delta})^2\Big).
\end{aligned}
\end{equation}
Similar to problem \eqref{sb1}, problem \eqref{sb21} is also a convex QCQP problem with only one constraint and strong duality holds for this problem. Therefore, it can be globally solved in closed-form and the details are presented in Appendix \ref{appendixD}.

\subsubsection{Group 2}  The Type 1 subproblem in group 2 can be obtained by changing the time slot indices in problem \eqref{sb1} from $2i$ and $2i+1$ to $2i+1$ and $2i+2$, respectively. Therefore, it can be solved by resorting to Appendix \ref{appendixB}, the details are not shown here for brevity. Similarly, the Type 2 subproblem can be obtained by changing the time slot indices in problem \eqref{sb3} and it can be efficiently solved according to Appendix \ref{appendixC}. Besides, the Type 3 subproblem is given by 
\begin{equation}
\label{sb22}
\begin{array}{l}
	\max\limits_{\{\hat{x}[2i],\; \hat{y}[2i]\}}\;
        L_{\delta,3}. \\
\end{array}
\end{equation}
Since problem \eqref{sb22} is an unconstrained convex problem, its global optimal solution can be easily obtained by (resorting to the first-order optimality condition)
\begin{equation}
\begin{aligned}
\hat{x}^{\textrm{opt}}[2i]=\frac{\ddot{x}[2i]+x[2i]}{2}+\frac{\theta_{x_{2i}}-\eta_{x_{2i}}}{2\delta},\;
\hat{y}^{\textrm{opt}}[2i]=\frac{\ddot{y}[2i]+y[2i]}{2}+\frac{\theta_{y_{2i}}-\eta_{y_{2i}}}{2\delta}.
\end{aligned}
\end{equation}

Finally, the dual variables can be updated by
\begin{equation}
\label{lagrange}
\begin{aligned}	
&\lambda_{x_i}=\lambda_{x_i}+\delta(\bar{x}[i]-x[i]),\; \lambda_{y_i}=\lambda_{y_i}+\delta(\bar{y}[i]-y[i]),\\
&\eta_{x_i}=\eta_{x_i}+\delta(\hat{x}[i]-x[i]),\; \eta_{y_i}=\eta_{y_i}+\delta(\hat{y}[i]-y[i]),\\
& \omega_{x_i}=\omega_{x_i}+\delta(\tilde{x}[i]-x[i]),\;\omega_{y_i}=\omega_{y_i}+\delta(\tilde{y}[i]-y[i]),\\
&\theta_{x_i}=\theta_{x_i}+\delta(\ddot{x}[i]-\hat{x}[i]),\; \theta_{y_i}=\theta_{y_i}+\delta(\ddot{y}[i]-\hat{y}[i]).
\end{aligned}
\end{equation}
Overall, the proposed algorithm to solve problem \eqref{tproblem_appro} is summarized in Algorithm \ref{tradesign}. Note that if the 3D trajectory optimization is considered, we can similarly introduce auxiliary variables for the altitudes of the UAV and solve the resulting subproblems accordingly without much difficulty.

\begin{algorithm}[!h] \small
  \centering
  \caption{Proposed ADMM Method for Problem \eqref{tproblem_appro}}
  \begin{algorithmic}[1] \label{tradesign}
 \STATE Let $\mathcal{U}=\mathbf{0}$, set a threshold $\epsilon$ and the penalty parameter $\delta$.
 \REPEAT 
 \STATE Update the variables in group 1 by solving subproblems \eqref{sb1}, \eqref{sb3} and \eqref{sb21}.
 \STATE Change the time slot indices in subproblems \eqref{sb1} and \eqref{sb3} and update the variables in group 2 by solving subproblems \eqref{sb1}, \eqref{sb3} and \eqref{sb22}.

 \STATE Update the dual variables according to \eqref{lagrange}

 \STATE Calculate the primal residual $\mathbf{r}$ and dual residual $\mathbf{s}$ using \eqref{prdr}.

 \UNTIL{$\max(\|\mathbf{r}\|,\|\mathbf{s}\|)<\epsilon$}.

  \STATE \textbf{output} $\{x[i], y[i]\}$.
 \end{algorithmic}
\end{algorithm}

\vspace{-1em}
\subsection{Overall Algorithm and Analysis} \label{analysis}
To summarize, the proposed algorithm can find a suboptimal solution of problem \eqref{traproblem} by applying the BCD method, i.e., the power allocation subproblem \eqref{pap} and the trajectory optimization subproblem \eqref{t_optimization_problem} are solved alternatively in an iterative manner. The detailed steps of the proposed algorithm are listed in Algorithm \ref{tab:table1}. Furthermore, regarding to the convergence of Algorithm \ref{tab:table1}, we have the following proposition: 
\newtheorem{proposition}{\underline{Proposition}}
\begin{proposition}
\emph{The sequence of the objective values generated by Algorithm \ref{tab:table1} is guaranteed to converge.}
\end{proposition}
\begin{proof}
Since problems \eqref{pap} and \eqref{t_optimization_problem} are equivalent to problems \eqref{tstpap} and \eqref{tproblem_equi}, respectively, and the latter two can be approximated by problems \eqref{pap1} and \eqref{tproblem_appro} through the first-order approximations, we can infer that the solution obtained in the $(t-1)$-th iteration of Algorithm \ref{tab:table1}, denoted by $\{p^{t-1}[i], \rho^{t-1}[i], x^{t-1}[i], y^{t-1}[i] \}$, is feasible to problem \eqref{traproblem}. Besides, due to the fact that Algorithm
\ref{tab:table2} and Algorithm \ref{tradesign} can obtain the optimal solutions of problems \eqref{pap1} and \eqref{tproblem_appro}, respectively, it can be readily seen that $\bar{R}_{\textrm{as}} (\{x^t[i],y^t[i],p^t[i],\rho^t[i] \})\geq \bar{R}_{\textrm{as}} (\{x^{t-1}[i],y^{t-1}[i],p^{t-1}[i],\rho^{t-1}[i] \})$. Together with the fact that the objective value of problem \eqref{traproblem} is upper bounded by a certain value due to the power constraints \eqref{avpower} and \eqref{peakpower}, we conclude that the sequence $\{\bar{R}_{\textrm{as}} (\{x^t[i],y^t[i],p^t[i],$ $\rho^t[i] \})\}$ guarantees to converge. This completes the proof.
\end{proof}

	\begin{algorithm}[!h] \small 
	\centering
	\caption{Proposed  Algorithm for Problem \eqref{traproblem}}
	\begin{algorithmic}[1] \label{tab:table1}
		\STATE Initialize $\mathcal{Q}$,  $\{p[i], \rho[i]\}$ and set a threshold $\tau$.
		\REPEAT
			\STATE Solve problem \eqref{pap1} using Algorithm \ref{tab:table2} with fixed trajectory and obtain $\{p[i], \rho[i]\}$.
		\STATE Solve problem \eqref{tproblem_appro} using Algorithm \ref{tradesign} with fixed $\{p[i],\rho[i]\}$ and obtain $\mathcal{Q}$.
		\STATE $\mathcal{Q}_f \gets \mathcal{Q}$, $\{p_f[i], \rho_f[i]\} \gets \{p[i], \rho[i]\}$.
		\UNTIL{The fractional increase of the objective value of problem \eqref{traproblem} is below the threshold $\tau$.}
		\STATE \textbf{output} $\{ p[i], \rho[i], x[i], y[i]\}$.
	\end{algorithmic}
\end{algorithm}

Besides, Algorithm \ref{tab:table1} exhibits very low computational complexity and the detailed analysis is presented as follows. As mentioned in Section \ref{subsection_pa}, since the power allocation subproblem is divided into $T$ independent subproblems and each subproblem is solved efficiently in closed-form, the worst-case complexity of Algorithm \ref{tab:table2} is $\mathcal{O}(N_BT)$, where $N_B$ denotes the number of iterations required by the Bisection method. For Algorithm \ref{tradesign}, we can see that its complexity
is dominated by solving problem \eqref{sb21} using Gaussian eliminations.  Since the complexity of solving one instance of problem \eqref{sb21} does not scale with $T$, the complexity of Algorithm \ref{tradesign} can be expressed as $\mathcal{O}(N_AN_BT)$, where $N_A$ denotes the number of ADMM iterations. In summary, the complexity of Algorithm \ref{tab:table1} can be expressed as $\mathcal{O}(N_{\textrm{BCD}}(N_AN_BT+N_BT))$, where $N_{\textrm{BCD}}$ represents the number of BCD iterations. Note that the complexity of the conventional algorithm in \cite{Zhang2019Secure} is $\mathcal{O}(N_{\textrm{BCD}} T^{3.5})$, therefore, the proposed Algorithm \ref{tab:table1} exhibits a much lower
complexity\footnote{Since $T$ is usually on the order of several hundreds, thus $T^{3.5} \gg N_A N_B T$.} and it will be shown in Section \ref{Section6:simulations} that Algorithm \ref{tab:table1} can achieve a similar performance with that of the conventional algorithm using existing convex solvers.

\vspace{-0.6em}
\section{ Simulation results}
\label{Section6:simulations}
In this section, we provide numerical results to evaluate the performance of our proposed low-complexity algorithm (i.e., Algorithm \ref{tab:table1}).  For comparison, we also provide the performance of the following three benchmark schemes:
\begin{itemize}
\item The fixed trajectory (FT) scheme: the transmit power levels and power splitting ratios are jointly optimized, while the UAV is assumed to fly from $(x_1,y_1)$ to $(x_T,y_T)$ straightly at a constant speed.
\item The naive power splitting (NPS) scheme: running Algorithm \ref{tab:table1} with fixed $\rho[i]=0.5,\;\forall i$.
\item The without AN scheme: running Algorithm \ref{tab:table1} with fixed $\rho[i]=1,\;\forall i$.
\end{itemize}
In our simulations, the channel bandwidth, the noise power spectrum and the channel power gain are set to $20$ MHz, $N_0=-169$ dBm/Hz and $\gamma_0=-36$ dB, respectively, and the carrier frequency is set at $5$ GHz. Hence, the reference SNR at a distance of $1$ m is $\frac{\gamma_0}{\sigma^2}=60$ dB.  The nominal system configuration is defined by the following choice of parameters unless otherwise specified: $L=100$ m, $H = 100$ m, $V_{\textrm{max}}=12$ m/s, $M = 4$ kg, $N=125$ s, $\delta_t=0.5$ s, $ (x_1,y_1)=(-200\;\textrm{m}, -150\;\textrm{m})$, $(x_T,y_T) = (1000\; \textrm{m}, -150\;\textrm{m})$, $\bar{P}=0$ dBm and $P_{\textrm{max}}=4\bar{P}$. The total mobility energy stored at the UAV $E_\textrm{tr}$ is set to $19.40$ kJ.

\subsubsection{Convergence property}
First, we illustrate in Fig. \ref{convergence} the convergence of our proposed Algorithm \ref{tradesign} and \ref{tab:table1}. From Fig. \ref{convergence} (a), it is observed that the outer BCD iteration of Algorithm \ref{tab:table1} is monotonically convergent and it needs about $10$ iterations to obtain the steady performance. Besides, in Fig. \ref{convergence} (b) and (c), we plot the primal and dual residuals $\|\mathbf{r}\|$ and $\|\mathbf{s}\|$ versus the number of ADMM iterations in Algorithm \ref{tradesign}. As can be seen, Algorithm \ref{tradesign} can converge well within $2000$ iterations. Although this number is relatively large as compared with the number of outer BCD iterations, the complexity is low since each updating step in Algorithm \ref{tradesign} is very simple, this will be verified in the following results.

\begin{figure}[htbp]
	\setlength{\abovecaptionskip}{-0.05cm}
	\setlength{\belowcaptionskip}{-0.05cm}
	\centering
	\includegraphics[scale=0.48]{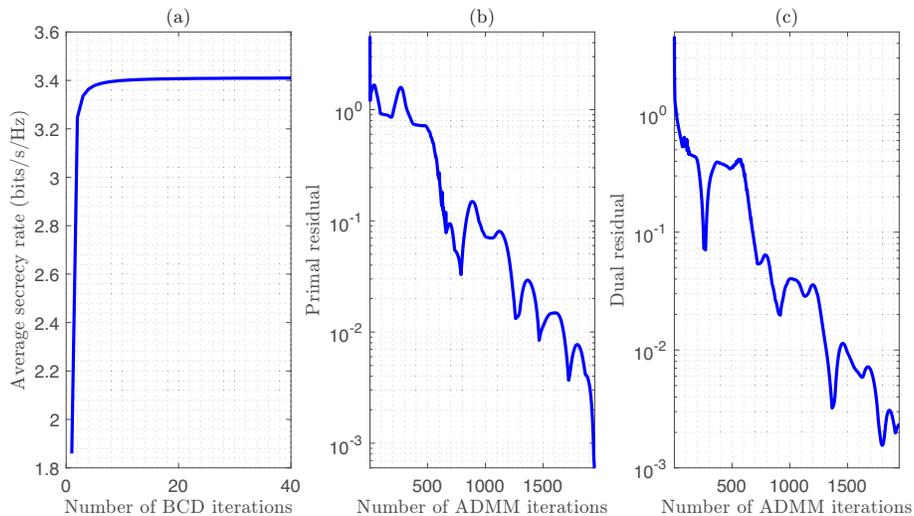}
	\caption{Convergence behavior of the proposed Algorithm \ref{tradesign} and \ref{tab:table1}.} \vspace{-1em}
	\label{convergence}
\end{figure}

\subsubsection{Performance and complexity}
In Fig. \ref{Fig6_running_time} and Table \ref{objvalue1}, we respectively investigate the average running time required by Algorithm \ref{tab:table1} to complete one outer BCD iteration and the achieved objective value (i.e., the average secrecy rate) in bits/s/Hz by Algorithm \ref{tab:table1}. For comparison, we also provide the performance achieved by replacing Algorithm \ref{tab:table2} and \ref{tradesign} in steps 3 and 4 of Algorithm \ref{tab:table1} by using the CVX solver \cite{cvx}. From Fig. \ref{Fig6_running_time}, we observe that the running time required by the proposed algorithm is significantly less than that required by using CVX. The running time increases with the increasing of $T$, however, it increases much
slower for the proposed algorithm. This is consistent with the complexity analysis in Section \ref{analysis} and it shows that the proposed algorithm design is more scalable. Besides, we observe from Table \ref{objvalue1} that the average secrecy rate achieved by the proposed algorithm and that by CVX is almost identical. In certain cases, such as $T=240$, the performance of the proposed algorithm is even better. This is because the CVX solver uses a successive approximation heuristic method to solve convex optimization problems involving $\log(\cdot)$ functions, which may lead to certain performance loss due to precision issues.

\begin{figure}[htbp]
	\setlength{\abovecaptionskip}{-0.05cm}
	\setlength{\belowcaptionskip}{-0.05cm}
	\centering
	\includegraphics[scale=0.48]{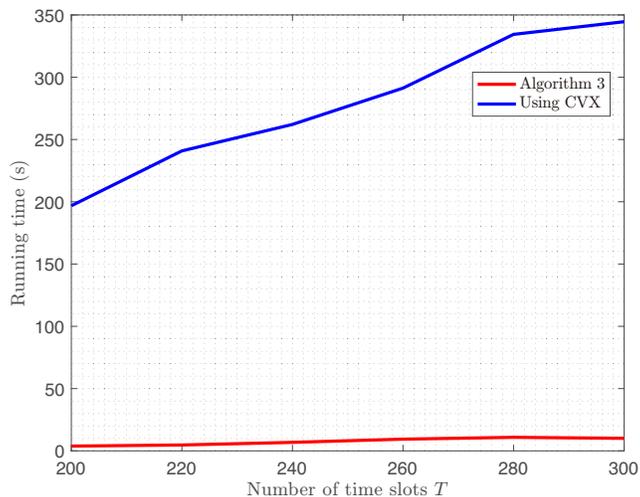}
	\caption{Running time comparison.} \vspace{-1em}
	\label{Fig6_running_time}
\end{figure}

	\begin{table} [htbp]
		\scriptsize
	\centering
	\caption{ Achieved Average Secrecy Rate Comparison} 
	\begin{tabular}{ccccccc}
		\toprule
		& \multicolumn{6}{c}{Numbers of time slots $T$} \\
		\cmidrule{2-7}
		& 200 & 220 & 240 & 260 & 280 & 300 \\
		\midrule
		Using CVX  & 2.2019	& 2.9532 & 3.2711 & 3.5375 & 3.7640 & 3.9580\\
		Algorithm \ref{tab:table1} & 2.2019 & 2.9532 & 3.2721 & 3.5375 & 3.7640 & 3.9580 \\
		\bottomrule
	\end{tabular}
	\label{objvalue1}
\end{table}

\subsubsection{Average secrecy rate versus the Bob-Eve distance $L$}
In Fig. \ref{Fig1_L}, we plot the average secrecy rates achieved by the considered schemes under various values of $L$. First, it is observed that the proposed algorithm achieves the best performance among the considered schemes. Second, the achieved average secrecy rates by all the considered schemes increases with $L$, which is expected since it is more difficult for Eve to intercept the communications between Bob and the UAV when $L$ is large. Similarly, since transmitting AN is less important under larger $L$, the performance of the without AN scheme approaches that of the proposed algorithm with the increasing of $L$. Besides, we observe that the performance of the NPS scheme is better than that of the FT
scheme. This is due to the fact that optimizing the UAV's trajectory under the considered simulation setup enables the UAV to fly close to Bob and away from Eve to achieve higher secrecy rate, while the performance gain offered by optimizing the power splitting ratios $\{\rho[i]\}$ is not that pronounced.

\begin{figure}[!htbp]
	\setlength{\abovecaptionskip}{-0.05cm}
	\setlength{\belowcaptionskip}{-0.05cm}
	\centering
	\includegraphics[scale=0.48]{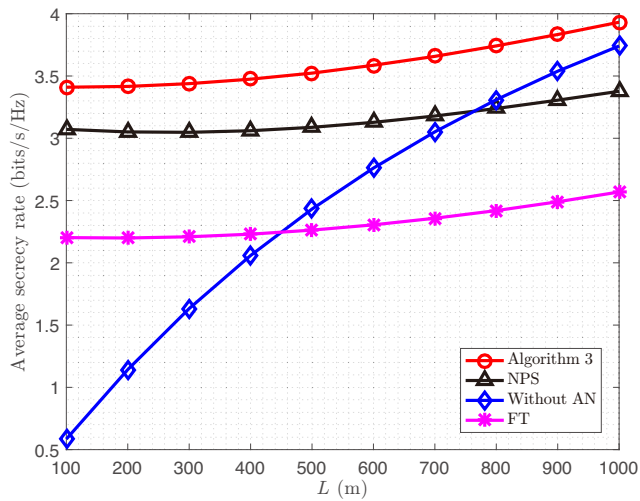}
	\caption{Average secrecy rate versus $L$.} \vspace{-1em}
	\label{Fig1_L}
\end{figure}

\begin{figure}[htbp]
	\setlength{\abovecaptionskip}{-0.05cm}
	\setlength{\belowcaptionskip}{-0.05cm}
	\centering
	\includegraphics[scale=0.48]{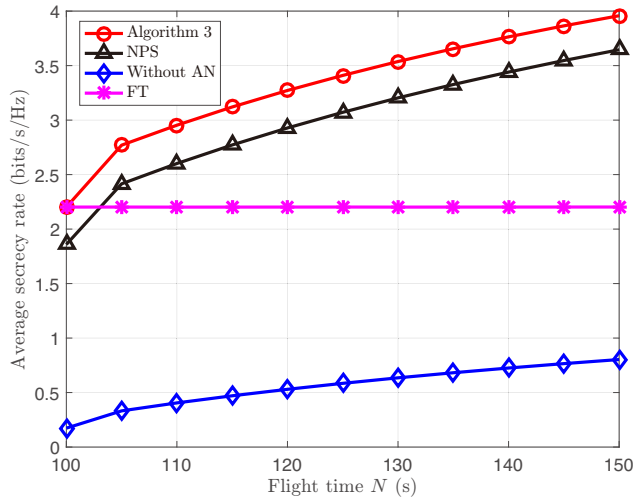}
	\caption{Average secrecy rate versus $N$.} \vspace{-1em}
	\label{SRVSTS}
\end{figure}

\subsubsection{Average secrecy rate versus the total flight time $N$}
In Fig. \ref{SRVSTS}, we investigate the average secrecy rate versus various values of $N$. As can be seen, the performance of all the considered schemes improves with the increasing of $N$, except for the FT scheme. This is because with increasingly large $T$, the UAV is able to hover over its favorable locations for a longer time, which leads to higher secrecy rate. However, if the mobility of the UAV cannot be exploited as in the FT scheme, the achieved average secrecy rate will remain unchanged even for sufficiently large $N$. Besides,
we can observe that when $N$ is small (e.g., $N=100$ s), the performance of the FT scheme is better than that of the NPS scheme, since in this case, the advantage of mobility control cannot be exploited due to the limited flight time. Moreover, the proposed Algorithm \ref{tab:table1} outperforms the other analyzed schemes.

\subsubsection{Average secrecy rate versus the average transmit power $\bar{P}$}
Fig. \ref{SRVSAVP} plots the average secrecy rates of different schemes versus $\bar{P}$. As shown, the proposed Algorithm \ref{tab:table1} always achieves the highest average secrecy rate, while the without AN scheme provides the lowest average secrecy rate. The performance achieved by the proposed scheme, the FT scheme and the NPS scheme all improves with the increasing of $\bar{P}$, while that by the without AN scheme does not change much. This is because the Bob-Eve distance is set to $L=100$ m, which is relatively
close and thus the qualities of the UAB-Bob and UAV-Eve links both improve as $\bar{P}$ increases since no AN is available to degrade the UAV-Eve link. Besides, we observe that the performance gain of the proposed algorithm over the NPS scheme gradually decreases and approaches zero as $\bar{P}$ increases. This is reasonable since the achievable rates of the UAV-Bob and UAV-Eve links are $\log(\cdot)$ functions of $\{\textrm{SNR}_{\textrm{I}}[i]\}$ and $\{ \textrm{SINR}_{\textrm{E}}[i] \}$, they tend to gradually saturate as $\bar{P}$ increases, which will limit the gain offered by the power splitting.

\begin{figure}[htbp]
		\setlength{\abovecaptionskip}{-0.05cm}
	\setlength{\belowcaptionskip}{-0.05cm}
	\centering
	\includegraphics[scale=0.48]{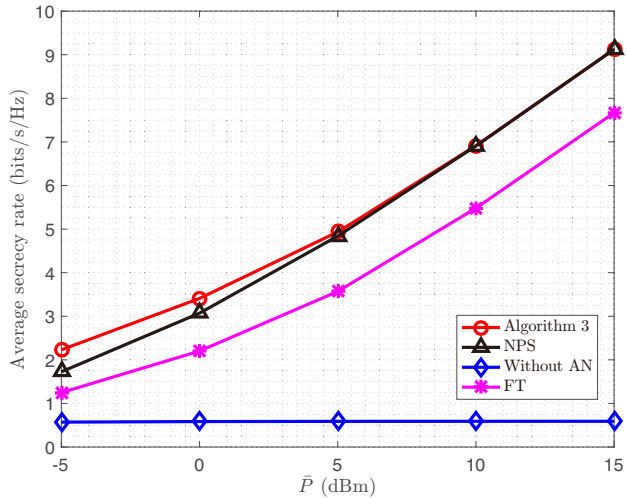}
	\caption{Average secrecy rate versus $\bar{P}$.} \vspace{-1em}
	\label{SRVSAVP}
\end{figure}

\subsubsection{Trajectories under various values of $E_\textrm{th}$}
Fig. \ref{trajectoryptr} shows the trajectories of the UAV by employing different schemes when $E_{\textrm{th}} = 14.55$ kJ and $E_{\textrm{th}} = 19.40$ kJ. First, we can see that with larger $E_{\textrm{th}}$, the UAV can fly closer to Bob to achieve a higher secrecy rate and this holds for all the considered schemes. Second, it is observed that the trajectories of the proposed algorithm and the without AN scheme differ significantly with $E_{\textrm{th}} = 14.55$ kJ or $E_{\textrm{th}} = 19.40$ kJ, especially when the UAV flies towards Bob.  Specifically, with the ability to transmit AN (in the proposed algorithm and the NPS scheme), the UAV can fly closer to Bob and Eve, while for the without AN scheme, the UAV has to keep a certain distance away from Bob in order to weaken the UAV-Eve link. Besides, the trajectories of the proposed algorithm and the NPS scheme are almost identical, which implies that the UAV's trajectory is not sensitive to the power splitting ratios under the considered simulation setups.

\begin{figure}[htbp] \vspace{-1em}
		\setlength{\abovecaptionskip}{-0.05cm}
	\setlength{\belowcaptionskip}{-0.05cm}
	\centering
	\includegraphics[scale=0.48]{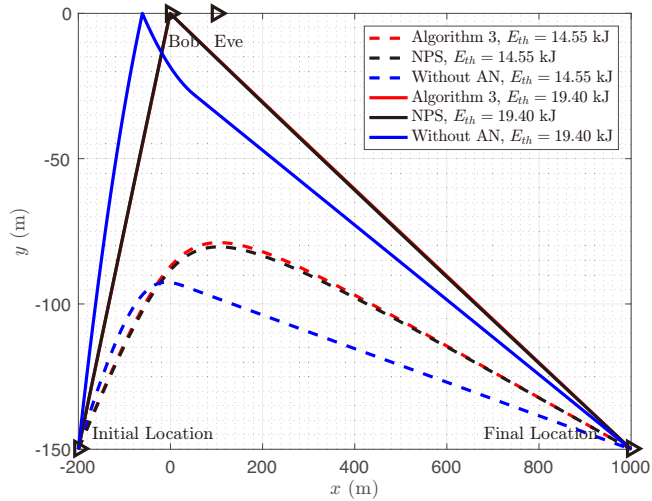}
	\caption{Trajectories of the UAV under various values of $E_\textrm{th}$.} \vspace{-0.5em}
	\label{trajectoryptr}
\end{figure}
	
\subsubsection{Trajectories under various values of $N$}
In Fig. \ref{TrajectoryFigure}, we show the trajectories of the UAV by employing different schemes when $N=103$ s and $N = 125$ s. We observe that when the flight time is long enough (i.e., $N=125$ s), the UAV can fly close to Bob and Eve, while when $N=103$ s, the UAV has to head back to the final location before it can reach its most favorable location. Besides, similar to the results in Fig. \ref{trajectoryptr}, the trajectories of the proposed algorithm and the without AN scheme are different owing to the difference in the ability of transmitting AN signals.
\vspace{-1em}
\begin{figure}[htbp]
		\setlength{\abovecaptionskip}{-0.05cm}
	\setlength{\belowcaptionskip}{-0.05cm}
	\centering
	\includegraphics[scale=0.48]{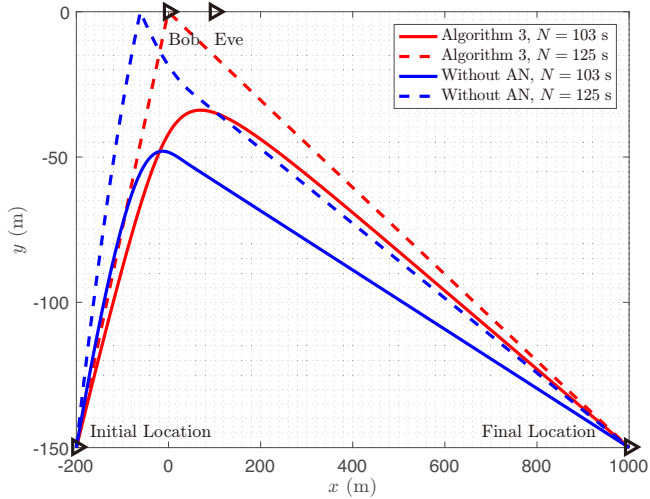}
	\caption{Trajectories of UAV using different algorithms.} \vspace{-1em}
	\label{TrajectoryFigure}
\end{figure}

\section{Conclusion}
\label{Section7:conclusion}
In this work, we proposed a power splitting approach to secure the UAV communication against a potential ground Eve, by enabling the UAV to transmit confidential information and AN simultaneously. By exploiting the power splitting capability of the UAV and its controllable mobility, we maximized the average secrecy rate by jointly optimizing the UAVs’ trajectory, the transmit power levels and the power splitting ratios over time. An iterative algorithm with very low complexity was proposed to solve the considered optimization problem with guaranteed convergence. Numerical results showed the effectiveness of our proposed algorithm. It is worth noting that the proposed algorithm and the underlying techniques that are employed can be extended to other joint power and trajectory optimization problems for UAV-enabled communication systems.
\vspace{-4mm}

\begin{appendices}
\section{Brief introduction to ADMM}
\label{appendixA}
To illustrate the idea of the ADMM, let us consider the following convex optimization problem:
\begin{equation} \label{ADMM}
\begin{aligned}
	\min\limits_{\mathbf{x}\in\mathbb{R}^{n \times 1},\; \mathbf{z}\in\mathbb{R}^{m \times 1}}\;  f(\mathbf{x})+g(\mathbf{z}) \quad
	\textrm{s.t.}\;  \mathbf{Ax}+\mathbf{Bz}=\mathbf{c},\;
	\mathbf{x} \in \mathcal{C}_1,\; \mathbf{z} \in \mathcal{C}_2,
	\end{aligned}
\end{equation}
where $f(\cdot): \mathbb{R}^{n\times 1}\mapsto \mathbb{R}$  and $g(\cdot): \mathbb{R}^{m\times 1} \mapsto \mathbb{R}$ are convex functions, $\mathcal{C}_1 \in \mathbb{R}^{n\times 1}$ and $\mathcal{C}_2 \in \mathbb{R}^{m\times 1}$ are non-empty convex sets, $\mathbf{A}\in\mathbb{R}^{p\times n},\mathbf{B}\in\mathbb{R}^{p\times m},\mathbf{c}\in\mathbb{R}^{p \times 1}$. Assume that problem \eqref{ADMM} is feasible and strong duality holds.

The ADMM solves problem \eqref{ADMM} by resorting to the following AL problem:
\begin{equation}
\begin{aligned}
\min\limits_{\mathbf{x}\in\mathbb{R}^{n \times 1},\;\mathbf{z}\in\mathbb{R}^{m \times 1}} \;  L_{\rho}(\mathbf{x},\mathbf{z},\bm{\lambda})\quad
\textrm{s.t.}\;  \mathbf{x} \in \mathcal{C}_1,\; \mathbf{z} \in \mathcal{C}_2,
\end{aligned}
\end{equation}
where $L_{\rho}(\mathbf{x},\mathbf{z},\bm{\lambda})= f(\mathbf{x})+g(\mathbf{z})+\bm{\lambda}^T(\mathbf{Ax}+\mathbf{Bz}-\mathbf{c})  +\frac{\rho}{2}\|\mathbf{Ax}+\mathbf{Bz}-\mathbf{c}\|^2$, $\bm{\lambda}$ denotes the dual variable and $\rho$ is the penalty parameter. Then, the ADMM iterates over the following three steps:
\begin{subequations}
\begin{align}
	&\mathbf{x}_{k+1}=\argmin_{\mathbf{x}\in\mathbb{R}^{n \times 1}} L_{\rho}(\mathbf{x},\mathbf{z}_k,\bm{\lambda}_k), \label{admm1}\\
	&\mathbf{z}_{k+1}=\argmin_{\mathbf{z}\in\mathbb{R}^{m \times 1}} L_{\rho}(\mathbf{x}_{k+1},\mathbf{z},\bm{\lambda}_k),\label{admm2}\\
	&\bm{\lambda}_{k+1}=\bm{\lambda}_k+\rho(\mathbf{Ax}_{k+1}+\mathbf{Bz}_{k+1}-\mathbf{c}),\label{admm3}
\end{align}
\end{subequations}
where $k$ denotes the iteration index. The convergence criterion of the ADMM can be expressed as $\|\mathbf{r}_{k+1}\|\leq\epsilon$ and $ \|\mathbf{s}_{k+1}\|\leq\epsilon$,  where $\mathbf{r}_{k+1}$ and $\mathbf{s}_{k+1}$ denote the primal residual and dual residual in the $(k+1)$-th iteration, which are defined as
\begin{equation}
\label{prdr}
\begin{aligned}
\mathbf{r}_{k+1}=\mathbf{Ax}_{k+1}+\mathbf{Bz}_{k+1}-\mathbf{c},\;
\mathbf{s}_{k+1}=\rho\mathbf{A}^T\mathbf{B}(\mathbf{z}_{k+1}-\mathbf{z}_{k}).
\end{aligned}
\end{equation}
It can be seen that the ADMM alternatively performs one iteration of primal variables updates, i.e., \eqref{admm1} and \eqref{admm2},  and one step of outer subgradient update for the dual variable, i.e., \eqref{admm3}. It converges to the global optimum of problem \eqref{ADMM} under relatively loose conditions. For more details, please refer to \cite{Bertsekas1989}.

\section{Optimal Solution to Problem \eqref{sb1}}
\label{appendixB}
It can be readily seen that problem \eqref{sb1} is convex and strong duality holds, therefore, it can be solved by resorting to the dual problem. Specifically, the Lagrangian function of problem  \eqref{sb1} is given by $L_1=L_{\delta,1}-\mu((x[2i]-\bar{x}[2i+1])^{2}+(y[2i]-\bar{y}[2i+1])^{2}- V_{\textrm{max}}^2)$, where $\mu$ denotes the Lagrangian multiplier. Then, by setting $\frac{\partial L_1}{\partial\bar{x}[2i+1]}=0$, we have
\begin{equation}
	\label{xbarnpo}
	\bar{x}^{\textrm{opt}}[2i+1]=\frac{1}{2\mu+\delta}(2\mu x[2i]+\delta x[2i+1]-\lambda_{x_{2i+1}}).
\end{equation}
Substituting \eqref{xbarnpo} into $L_1$ and taking the partial derivative of $L_1$ with respect to $x[2i]$, we can obtain
\begin{equation} \label{appendix_b_1}
\begin{aligned}
	& -\delta\big(x[2i]-\bar{x}[2i]-\frac{\lambda_{x_{2i}}}{\delta}+x[2i]-\hat{x}[2i] -\frac{\eta_{x_{2i}}}{\delta}+x[2i] \\
	&  -\tilde{x}[2i]-\frac{\omega_{x_{2i}}}{\delta}\big) -\frac{2\mu\delta}{2\mu+\delta}\big(x[2i] -x[2x+1] +\frac{\lambda_{x_{2i+1}}}{\delta}\big)=0.
\end{aligned}
\end{equation}
Based on \eqref{appendix_b_1}, the optimal $x^{\textrm{opt}}[2i]$ can be expressed as
\begin{equation}
\begin{aligned}
	\label{xn}
	x^{\textrm{opt}}[2i]= \frac{1}{3\delta+\frac{2\mu\delta}{2\mu+\delta}}\big(\delta(\bar{x}[2i]+\hat{x}[2i]+\tilde{x}[2i])+\lambda_{x_{2i}} +\eta_{x_{2i}}+\omega_{x_{2i}}+\frac{2\mu\delta}{2\mu+\delta}(x[2i+1]-\frac{\lambda_{x_{2i+1}}}{\delta})\big).
\end{aligned}
\end{equation}
Similarly, we have
\begin{equation}
	\label{ybarnpo}
	\bar{y}^{\textrm{opt}}[2i+1]=(2\mu y[2i]+\delta y[2i+1]-\lambda_{y_{2i+1}})/(2\mu+\delta),
\end{equation}
and
\begin{equation}
	\label{yn}
	\begin{aligned}
	y^{\textrm{opt}}[2i]= \frac{1}{3\delta+\frac{2\mu\delta}{2\mu+\delta}}\big(\delta(\bar{y}[2i]+\hat{y}[2i]+\tilde{y}[2i])+\lambda_{y_{2i}} +\eta_{y_{2i}}+\omega_{y_{2i}}+\frac{2\mu\delta}{2\mu+\delta}(y[2i+1]-\frac{\lambda_{y_{2i+1}}}{\delta})\big).
	\end{aligned}
\end{equation}
Then,  it is not difficult to see that if $\bar{x}^{\textrm{opt}}[2i+1]$, $x^{\textrm{opt}}[2i]$, $\bar{y}^{\textrm{opt}}[2i+1]$ and $y^{\textrm{opt}}[2i]$ satisfy $(x^{\textrm{opt}}[2i]-\bar{x}^{\textrm{opt}}[2i+1])^{2}+(y^{\textrm{opt}}[2i]-\bar{y}^{\textrm{opt}}[2i+1])^{2}\leq  V_{\textrm{max}}^2$ when $\mu=0$, then this is the optimal solution. Otherwise, we substitute \eqref{xbarnpo}, \eqref{xn}, \eqref{ybarnpo} and \eqref{yn} into $(x[2i]-\bar{x}[2i+1])^{2}+(y[2i]-\bar{y}[2i+1])^{2}= V_{\textrm{max}}^2$ (due to the complementary slackness). By solving this equation with respect to $\mu$, we have $\mu^{\textrm{opt}}={(\sqrt{A}-3\delta^2)}/{(8\delta)}$, where
\begin{equation}
\begin{aligned}
A=& \frac{\delta^2}{V_{\textrm{max}}^2}\Big(\big(\delta(\hat{x}[2i]+\tilde{x}[2i]+\bar{x}[2i])+\omega_{x_{2i}}+\eta_{x_{2i}}+\lambda_{x_{2i}}+3\lambda_{x_{2i+1}}-3\delta x[2i+1]\big)^2\\
& +\big(\delta(\hat{y}[2i]+\tilde{y}[2i]+\bar{y}[2i])+\omega_{y_{2i}} +\eta_{y_{2i}}+\lambda_{y_{2i}}+3\lambda_{y_{2i+1}}-3\delta y[2i+1]\big)^2\Big).
\end{aligned}
\end{equation}
Substituting $\mu^{\textrm{opt}}$ back into \eqref{xbarnpo}, \eqref{xn}, \eqref{ybarnpo} and \eqref{yn}, we can obtain the optimal solution of problem \eqref{sb1}.

\section{Optimal Solution to Problem \eqref{sb3}} \label{appendixC}
For notational simplicity, in this appendix, we ignore the time slot index $2i+1$ in the variables $\{\tilde{x}[2i+1],\tilde{y}[2i+1],x[2i+1], y[2i+1], u[2i+1], t[2i+1],\rho[2i+1], p[2i+1], \omega_{x_{2i+1}}, \omega_{y_{2i+1}}\}$ without loss of generality. First, we can observe that $L_{\delta,2}$ is a decreasing function with respect to $u$, therefore the optimal $u$, denoted as $u^{\textrm{opt}}$, must satisfy $u^{\textrm{opt}}=\tilde{x}^2+\tilde{y}^2+H^2$.  By substituting $u^{\textrm{opt}}$ into $L_{\delta,2}$, we obtain $L_{\delta,2} = -a(\tilde{x}^2+\tilde{y}^2+H^2) -\frac{\sigma^2t}{t_f+\gamma_0p}+\log(\gamma_0(1-\rho)p+\sigma^2t) -\frac{\delta}{2}\big((x-\tilde{x}-\frac{\omega_{x}}{\delta})^2 +(y-\tilde{y}-\frac{\omega_{y}}{\delta})^2\big)$.  Hence, problem \eqref{sb3} becomes
\begin{equation} \label{tsb3}
\max\limits_{\tilde{x},\; \tilde{y},\; t}\; L_{\delta,2} \quad
\textrm{s.t.} \; t\leq -\tilde{x}_f^2+2\tilde{x}_f\tilde{x}+L^2
-2\tilde{x}L-\tilde{y}_f^2+2\tilde{y}_f\tilde{y}+H^2.
\end{equation}
Since problem \eqref{tsb3} is convex, we can globally solve it by resorting to its dual problem. The corresponding Lagrange function for problem \eqref{tsb3} can be expressed as $L_{\delta,2,\tilde{\mu}}\triangleq L_{\delta,2} - \tilde{\mu}\big( t+\tilde{x}_f^2 -2(\tilde{x}_f-L)\tilde{x}-L^2 +\tilde{y}_f^2-2\tilde{y}_f\tilde{y}-H^2\big)$, where $\tilde{\mu}$ is the dual variable.

By checking the first-order optimality condition, we can express the optimal solution of problem \eqref{tsb3} as a function of $\tilde{\mu}$, i.e., 
\begin{equation}
\label{tsb3sm}
\begin{aligned}
& \tilde{x}^{\textrm{opt}}(\tilde{\mu})=\frac{\delta x-\omega_{x} + 2\tilde{\mu}(\tilde{x}_f-L)}{2a+\delta},\;
\tilde{y}^{\textrm{opt}}(\tilde{\mu})=\frac{\delta y-\omega_{y} + 2\tilde{\mu}\tilde{y}_f}{2a+\delta},\\
& t^{\textrm{opt}}(\tilde{\mu})=\frac{\sigma^2t_f+\gamma_0p}{\tilde{\mu}(\sigma^2t_f+\gamma_0p)+{\sigma^2}}-\frac{\gamma_0(1-\rho)p}{\sigma^2}.
\end{aligned}
\end{equation}
If the solution $\{\tilde{x}^{\textrm{opt}}(0),\tilde{y}^{\textrm{opt}}(0),t^{\textrm{opt}}(0)\}$ automatically satisfies the constraint of problem \eqref{tsb3}, then it is the optimal solution, otherwise, we can see that the optimal dual variable $\tilde{\mu}^{\textrm{opt}}$ satisfies
\begin{equation}
\label{Lagproblem}
\begin{aligned}
	t^{\textrm{opt}}(\tilde{\mu}^{\textrm{opt}})= -\tilde{x}_f^2 +2(\tilde{x}_f-L)\tilde{x}^{\textrm{opt}}(\tilde{\mu}^{\textrm{opt}})+L^2  -\tilde{y}_f^2+2\tilde{y}_f\tilde{y}^{\textrm{opt}}(\tilde{\mu}^{\textrm{opt}})+H^2.
\end{aligned}
\end{equation}
Substituting \eqref{tsb3sm} into \eqref{Lagproblem} and solving the resulting quadratic equation, we obtain $\tilde{\mu}^{\textrm{opt}} = {\left(-b_{\tilde{\mu}}+\sqrt{b_{\tilde{\mu}}^2-4a_{\tilde{\mu}}c_{\tilde{\mu}}} \right) }/{(2a_{\tilde{\mu}})}$, where  $ a_{\tilde{\mu}}=\frac{4(\tilde{x}_f-L)^2+4\tilde{y}_f^2}{2a+\delta}$, $ b_{\tilde{\mu}}=\frac{\sigma^2 a_{\tilde{\mu}}}{\sigma^2t_f+\gamma_0p}+d_{\tilde{\mu}}$, $c_{\tilde{\mu}}= \frac{\sigma^2d_{\tilde{\mu}}}{\sigma^2t_f+\gamma_0p}-1$ and $d_{\tilde{\mu}}= -\tilde{x}_f^2-\tilde{y}_f^2+L^2+H^2+\frac{\gamma_0(1-\rho)p}{\sigma^2}  +\frac{2(\tilde{x}_f-L)(\delta x-\omega_{x})+2\tilde{y}_f(\delta y-\omega_{y})}{2a+\delta}$.

\section{Optimal Solution to Problem \eqref{sb21}} \label{appendixD}
The Lagrangian function of problem \eqref{sb21} can be expressed as $L_3= L_{\delta,3}-\phi\sum\limits_{i=1}^{T-1}((\ddot{x}[i]-\ddot{x}[i+1])^2+(\ddot{y}[i]-\ddot{y}[i+1])^2)+\phi\frac{E_{\textrm{tr}}}{\kappa}$. By setting $\frac{\partial L_{\delta,3}}{\partial\hat{x}[2i+1]}=0$ and $\frac{\partial L_{\delta,3}}{\partial\hat{y}[2i+1]}=0$, we have
\begin{equation}
	\label{xhatp}
	\hat{x}^{\textrm{opt}}[2i+1]=\frac{1}{2}\Big(\ddot{x}[2i+1]+\frac{\theta_{x_{2i+1}}}{\delta}+x[2i+1]-\frac{\eta_{x_{2i+1}}}{\delta}\Big),
\end{equation}
\begin{equation}
	\label{yhatp}
	\hat{y}^{\textrm{opt}}[2i+1]=\frac{1}{2}\Big(\ddot{y}[2i+1]+\frac{\theta_{y_{2i+1}}}{\delta}+y[2i+1]-\frac{\eta_{y_{2i+1}}}{\delta}\Big).
\end{equation}
Substituting \eqref{xhatp} and \eqref{yhatp} into $L_3$ and letting $\frac{\partial L_3}{\partial\ddot{x}[2i+1]}=0$ and $\frac{\partial L_3}{\partial \ddot{x}[2i]}=0$, we can obtain the following equations:
\begin{equation}
\begin{aligned}
	\left(-4\phi-\frac{\delta}{2}\right)\ddot{x}[2i+1]+2\phi(\ddot{x}[2i+2]+\ddot{x}[2i])
	+\frac{\delta}{2}\left(x[2i+1]-\frac{\eta_{x_{2i+1}}+\theta_{x_{2i+1}}}{\delta}\right)=0,
	\end{aligned}
\end{equation}
\begin{equation}
	-(\delta+4\phi)\ddot{x}[2i]+\delta\hat{x}[2i]-\theta_{x_{2i}}+2\phi(\ddot{x}[2i+1]+\ddot{x}[2i+1])=0.
\end{equation}
Similarly, for $\ddot{y}[2i+1]$ and $\ddot{y}[2i]$, we have
\begin{equation}
\begin{aligned}
	\left(-4\phi-\frac{\delta}{2}\right)\ddot{y}[2i+1]+2\phi(\ddot{y}[2i+2]+\ddot{y}[2i])
	+\frac{\delta}{2}\left(y[2i+1]-\frac{\eta_{y_{2i+1}}+\theta_{y_{2i+1}}}{\delta}\right)=0,
	\end{aligned}
\end{equation}
\begin{equation}
	-(\delta+4\phi)\ddot{y}[2i]+\delta\hat{y}[2i]-\theta_{y_{2i}}+2\phi(\ddot{y}[2i+1]+\ddot{y}[2i+1])=0.
\end{equation}
Together with $	\ddot{x}[1]=x_1$, $\ddot{x}[T]=x_T$, $\ddot{y}[1]=y_1$ and $\ddot{y}[T]=y_T$, we can employ the Gaussian elimination to solve the above equations for a given dual variable $\phi$ and the optimal $\phi$ can be found by using the Bisection method. 

\end{appendices}

\vspace{-1em}
\bibliographystyle{IEEEtran}
\bibliography{references}

\end{document}